\documentclass{article}
\usepackage{fullpage}
\usepackage{amsmath}
\usepackage{amsthm}
\usepackage{amssymb}
\usepackage{url}
\usepackage{graphicx}
\usepackage{verbatim}
\usepackage{url}
\usepackage{graphicx}
\usepackage{tikz}
\usetikzlibrary{calc}
\usetikzlibrary{patterns}
\tikzstyle{mybox} = [draw=black, fill=white,  thick,
    rectangle, inner sep=10pt, inner ysep=20pt]
\tikzstyle{mybox} = [draw=black, fill=white,  thick,
    rectangle, inner sep=2pt, inner ysep=2pt]

\newtheorem{thm}{Theorem}
\newtheorem{lemma}{Lemma}
\newtheorem{cor}{Corollary}
\newtheorem{prop}{Proposition}
\theoremstyle{definition}
\newtheorem{definition}{Definition}
\newtheorem{remark}{Remark}

\begin{document}

\title{An Algorithmic Separating Hyperplane Theorem and Its Applications}
\author{Bahman Kalantari \\
Department of Computer Science, Rutgers University, NJ\\
kalantari@cs.rutgers.edu
}
\date{}
\maketitle

\begin{abstract}
We first prove a new separating hyperplane theorem characterizing when a pair of compact convex subsets $K, K'$ of the Euclidean space intersect, and when they are disjoint. The theorem is distinct from classical separation theorems. It generalizes the {\it distance duality} proved in our earlier work for
testing the membership of a distinguished point in the convex hull of a finite point set. Next by utilizing the theorem, we develop a substantially generalized and stronger version of the  {\it Triangle Algorithm} introduced in the previous work  to perform any of the following three tasks: (1) To compute  a pair $(p,p') \in K \times K'$, where either the Euclidean distance $d(p,p')$ is to within a prescribed tolerance, or the orthogonal bisecting hyperplane of the line segment $pp'$ separates the two sets; (2)
When $K$ and $K'$ are disjoint, to compute  $(p,p') \in K \times K'$ so that $d(p,p')$ approximates $d(K,K')$ to within a prescribed tolerance; (3) When $K$ and $K'$ are disjoint, to compute a pair of parallel supporting hyperplanes $H,H'$ so that $d(H,H')$ is to within a prescribed tolerance of the optimal margin.  The worst-case complexity of each iteration is solving a linear objective over $K$ or $K'$. The resulting algorithm is a fully polynomial-time approximation scheme for such important special cases as when $K$ and $K'$ are convex hulls of  finite points sets, or the intersection of a finite number of halfspaces. The results find many theoretical and practical applications, such as in machine learning, statistics, linear, quadratic and convex programming. In particular,  in a separate article we report on a comparison of the Triangle Algorithm and SMO for solving the hard margin problem.
In future work we extend the applications to combinatorial and NP-complete problems.
\end{abstract}

{\bf Keywords:} Convex Sets, Separating Hyperplane Theorem, Convex Hull, Linear Programming, Quadratic Programming, Duality, Approximation Algorithms, Support Vector Machines, Statistics


\section{Introduction} Quoting  Rockafellar  on separation theorems in \cite{Rock}, ``{\it The notion of separation has proved to be one of the most fertile notions in convexity theory and its applications.}''  The separating hyperplane theorem,  stated in numerous books,  is one of the best known theorems in the theory of convexity and convex programming with numerous applications in optimization, operations research, business and economics.   There are several different versions of the theorem. Special cases of the theorem such as Farkas Lemma play a fundamental role in linear programming. In particular, the lemma gives rise to the LP duality theory.

In this article  we are interested in the separation of two convex subsets $K$ and $K'$ of $\mathbb{R} ^m$, assumed to be compact. We prove a new separating hyperplane theorem and make use of it to
present a conceptually simple algorithm that  computes
$(p, p') \in K \times K'$ where, either the Euclidean distance $d(p, p')$ is as small as we please, or the orthogonal bisecting hyperplane to the line segment $pp'$ separates $K$ and $K'$, or $d(p, p')$ approximates $d(K,K')$, as well as computing a pair of parallel supporting hyperplanes $H,H'$ such that $d(H,H')$ is as close to the optimal margin as desired. In particular, in contrast with numerous existential proofs of the separating hyperplane theorem, our proof is an  algorithmic proof of this fundamental theorem, offering many practical applications.

The work in this article generalizes our previous results in \cite{kal14}, developed for the special case when $K=conv(\{v_1, \dots, v_n\})$ (the convex hull of a finite point set) and  $K'=\{p'\}$, a singleton point. We refer to this special case as the {\it convex hull membership problem} (or {\it convex hull decision problem}).  We have,  $p' \in K$, if and only if
\begin{equation} \label{eq3}
p'=\sum_{i=1}^n \alpha v_i, \quad \sum_{i=1}^n \alpha_i=1, \quad \alpha_i \geq 0 \quad \forall i.
\end{equation}

Despite its simplicity, this special case is a fundamental problem in computational geometry and linear programming and finds applications in statistics, approximation theory, and machine learning.  From the theoretical point of view, the convex hull membership problem is solvable in polynomial time, e.g. via the pioneering algorithm of Khachiyan  \cite{kha79}, or  Karmarkar \cite{kar84}. Indeed a general linear programming problem can be reduced to this special case, see e.g. \cite{kha90}, \cite{jinkal}. For  large-scale problems however, greedy algorithms are preferable to polynomial-time algorithms.  The Frank-Wolfe algorithm \cite{Frank}, Gilbert's algorithm \cite{Gilbert}, and {\it sparse greedy approximation} are such algorithms. For connections between these  see  Clarkson \cite{clark2008},  G{\"a}rtner and Jaggi \cite{Gartner}. A problem closely related to the convex hull membership problem is to compute the distance from $p'$ to $K$.

A more general case is when $K=conv(\{v_1, \dots, v_n\})$ and
$K'=conv(\{v'_1, \dots, v'_{n'}\})$. Testing if $K$ and $K'$ intersect is identical with testing if their Minkowski difference, $K-K'$, contains the origin. It is easy to show that
$K-K'=conv(\{v_i-v'_j: v_j \in K, v'_j \in K'\})$. Thus the case of two convex hulls can be reduced to the convex hull membership problem. However, via this formulation, the number of points is $nn'$ and as we shall see it is more efficient to test if they intersect directly without this reduction.  Applications of the problem of  testing if such convex hulls intersect, and their separation  include, e.g. support vector machines (SVM) and the approximation of a function as convex combination of other functions, see e.g. Clarkson \cite{clark2008} and Zhang \cite{zhang},
Burges \cite{Burg} and \cite{kal14}.

According to the classical separating hyperplane theorem, $K$ and $K'$ are disjoint if and only if they can be separated by a hyperplane, i.e. if there exists $h \in \mathbb{R} ^m$, and
$a \in \mathbb{R}$ such that

\begin{equation} \label{eq1}
h^Tx < a, \quad  \forall x \in K, \quad h^Tx > a, \quad \forall x \in K'.
\end{equation}
The hyperplane
\begin{equation} \label{eq2}
H=\{x \in \mathbb{R} ^m:  \quad h^Tx = a \}
\end{equation}
separates $K$ and $K'$.

Standard proofs rely on the fact that the  minimum of  the Euclidean distance function $d(x,x')$ is attained and is positive:
\begin{equation}
\min \{d(x,x') :  x \in K, \quad x' \in K' \} >0.
\end{equation}
This approach is discussed in standard convex programming and nonlinear optimization books, such as \cite{Boyd} and \cite{Bazaraa}.  The formulation as an optimization of distance between a pair of convex sets, or its reduction to computing the distance of the origin from the Minkowski difference does not necessarily lend itself to a working or practical algorithm.   When the convex sets are described by linear or nonlinear inequalities, a Lagrangian duality can be stated, see e.g. \cite{Boyd}.   However, such approaches do not necessarily give rise to a practical algorithm, even in the special case of the convex hull membership problem, i.e. testing intersection or separation when  $K=conv(\{v_1, \dots, v_n\})$  and  $K'=\{p'\}$.

Our goal in this article is to give a new theory and algorithms for testing the intersection or separation of two compact convex sets.   In \cite{kal14} we  studied the special case of the convex hull membership problem, proving a {\it distance duality} theorem and then using it we described a very simple geometric algorithm, called {\it Triangle Algorithm} that either produces a point
$p \in K$ such that $d(p, p')$ is to within a prescribed tolerance, or a point $p \in K$  such that the orthogonal bisecting hyperplane to the line segment $pp'$ separates $p'$ from $K$.  Equivalently, in this case $K$ is contained in
$V(p)=\{x: d(x,p) < d(x, p')\}$, the {\it Voronoi cell} of $p$, and $V(p)$  excludes $p'$, see Figure \ref{Fig1}.  In fact in this case $d(p,p')$ gives a good approximation to   $\delta_*=d(p',K)=\min\{d(x,p'): x \in K \}$:
\begin{equation} \label{eq5}
\frac{1}{2} d(p,p') \leq \delta_* \leq d(p,p').
\end{equation}

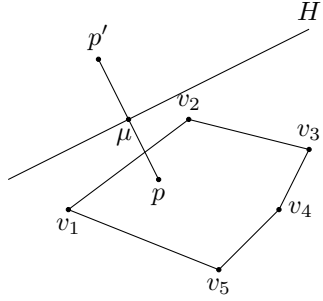
\begin{figure}[htpb]
	\centering
	
	\begin{tikzpicture}[scale=0.4]
			
\draw (0.0,0.0) -- (4,3.0) -- (8,2.0) --(7,0)--(5,-2)-- cycle;
		\draw (0,0) node[below] {$v_1$};
		\draw (7,0) node[right] {$v_4$};
		\draw (4,3) node[above] {$v_2$};
		\draw (5,-2) node[below] {$v_5$};
		\draw (8,2) node[above] {$v_3$};
\filldraw (5,-2) circle (2pt);
\filldraw (0,0) circle (2pt);
\filldraw (8,2) circle (2pt);
\filldraw (7,0) circle (2pt);
\filldraw (4,3) circle (2pt);
\filldraw (1,5) circle (2pt) node[above] {$p'$};
\filldraw (3,1) circle (2pt) node[below] {$p$};
\filldraw (2,3) circle (2pt) node[below] {$\mu~$};
\begin{scope}[black]
\draw (-2,1) -- (8,6);
\end{scope}
\draw (8,6) node[above] {$H$};
\draw (1,5) -- (3,1);
	

	\end{tikzpicture}
	
	\caption{Example of a case where orthogonal bisector of $pp'$ separates $K=conv(v_1, \dots, v_5)$ from $K'=\{p'\}$.}
	\label{Fig1}
\end{figure}

Based on preliminary experiments for solving the convex hull membership problem, the triangle algorithm performs quite well on reasonably large size problems, see \cite{Meng}. It can also be applied to solving linear systems, see \cite{kal12a} and \cite{Gibson} (for experimental results). Additionally, it can be applied to linear programming, see \cite{kal14}. Some variations of the Triangle Algorithm for the convex hull membership problem are given in \cite{Kalan12} and \cite{kalSaks}. Randomized versions of the algorithm, one inspired by the chaos game (see \cite{Barn93} and  \cite{Devaney2004}), are described in \cite{kalRand}. In view of these we anticipate that the Triangle Algorithms  will find practical applications in  distinct areas, including applications dealing with big data.

In the remainder of this section  we give a detailed outline of what is to follow in the subsequent sections.   In Section 2, we prove a new separating hyperplane theorem for compact convex sets. In Section 3, we prove a theorem that allows us to iteratively improve the distance between two convex sets. In Section 4, we describe an algorithm for testing if two convex sets intersect and if not it generates a separating hyperplane. In Section 5, we formally describe it as {\it Triangle Algorithm I} and analyze its complexity. In Section 6, we describe an algorithm for approximating the distance between two convex sets when they are proven to be disjoint, as well as approximating optimal parallel supporting hyperplanes.  In Section 7, we formally describe the latter algorithm as {\it Triangle Algorithm II} and analyze its complexity. In Section 8, we consider the complexity of the algorithms in several important special cases. In Section 9, we make concluding remarks and describe future work.

\subsection{Outline}

In this article we establish the following results, substantially generalizing the results in \cite{kal14}:

\textbf{(i)} We prove a version of the separating hyperplane theorem for the case where $K$ and $K'$ are arbitrary compact convex sets, not only giving a stronger version of the ordinary separating hyperplane theorem, but an algorithmic version which finds several practical applications.

\textbf{(ii)} We describe {\it Triangle Algorithm I} having the following properties:
Starting with a given pair $(p_0, p'_0) \in K \times K'$,
it either computes $(p, p') \in K \times K'$ such that $d(p,p')$ is
to within a prescribed tolerance, proving that $K$ and $K'$ are approximately intersecting,  or such that the orthogonal bisecting hyperplane of the line segment $pp'$ separates $K$ and $K'$, hence proving they are disjoint. We call such a pair $(p,p')$ a \emph{witness pair}.

\textbf{(iii) } We describe {\it Triangle Algorithm II} having the following properties: It begins with a witness pair $(p, p') \in K \times K'$, then it computes a new pair $(p, p') \in K \times K'$ such that $d(p,p')$ is to within a prescribed tolerance of the distance between $K$ and $K'$, $d(K,K')$. Then using this pair, it computes a pair of supporting hyperplanes $(H,H')$ parallel to the orthogonal bisecting hyperplane  of the line segment $pp'$, where $d(H,H')$ approximates the optimal margin to within a prescribed tolerance.

\textbf{(iv)} We analyze the complexity of Triangle Algorithms I and II for important special cases:

\begin{itemize}

\item

When  $K= conv(V)$, $V=\{v_1, \dots, v_n\}$, a subset of $\mathbb{R} ^m$, $K'= \{p'\}$, a singleton point in $\mathbb{R} ^m$. In particular,  this problem includes linear programming.

\item

When  $K= conv(V)$, $V=\{v_1, \dots, v_n\}$, $K'= conv(V')$, $V'=\{v_1', \dots, v_{n'}'\}$, $V, V' \subset \mathbb{R} ^m$. In particular, this has applications in machine learning such as SVM, see \cite{Vapnik1}, \cite{Vapnik2}

\item
When $K= \{x: Ax \leq b\}$,  $K'=\{x: A'x \leq b'\}$, where $A$ is $n \times m$ and $A'$ is $n' \times m$. In particular, when one set is a single point this includes such problems as strict convex quadratic programming.

\end{itemize}

To describe the complexity of these algorithm we need to give several definitions.  Assume we are given $p_0 \in K$, $p'_0 \in K'$.
Let
\begin{equation} \label{eqa4}
\delta_*=d(K,K')= \min \{d(p,p'): p \in K, p' \in K' \}.
\end{equation}
It is trivial to prove $\delta_* = 0$ if and only if $K \cap K' \not = \emptyset$.

\begin{definition} \label{def1} Suppose $\delta_*=0$. We say a pair $(p,p') \in K \times K'$ is an $\epsilon$-{\it approximation solution to the intersection problem} if
\begin{equation}
d(p,p') \leq \epsilon d(p,v), \quad \text{for some} \quad  v \in K, \quad \text{or} \quad  d(p,p') \leq \epsilon d(p',v'), \quad \text{for some} \quad v' \in K'.
\end{equation}
\end{definition}

\begin{definition} \label{def2} Given $(p,p') \in K \times K'$, we say it is a \emph{witness pair}
if the orthogonal bisecting hyperplane of the line segment $pp'$ separates $K$ and $K'$.
\end{definition}

\begin{definition} \label{def3} Suppose $\delta_*>0$. We say a witness pair $(p,p') \in K \times K'$ is an $\epsilon$-{\it approximation solution to the distance problem} (or $\epsilon$-{\it approximation solution to} $\delta_*$) if
\begin{equation}
d(p,p') - \delta_* \leq \epsilon d(p,p').
\end{equation}
\end{definition}

\begin{definition} \label{def4} Suppose $\delta_*>0$. We say  a pair of  parallel hyperplanes $(H,H')$ {\it supports} $(K,K')$, if $H$ contains a boundary point of $K$, $H'$ contains a boundary point of $K'$,
$K \subset H_+$, $K' \subset H'_+$, where $H_+, H'_+$  are disjoint halfspaces corresponding to $H,H'$.
\end{definition}

As an example, consider the case where $K$ and $K'$ are disjoint discs in the Euclidean plane. We can draw infinitely many supporting lines. These are parallel lines tangential to the discs, each touching the corresponding disc in exactly one point. The corresponding halfspaces separate $K$ and $K'$. See Figure \ref{Fig3X}

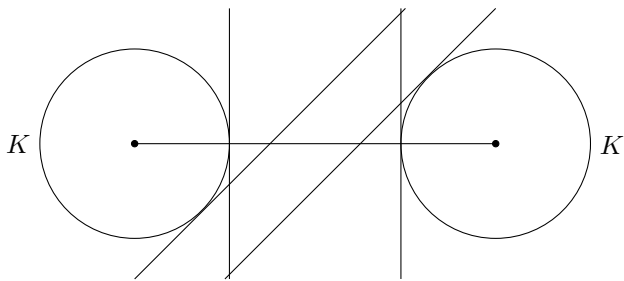
\begin{figure}[htpb]
	\centering
	\begin{tikzpicture}[scale=0.6]	
      \draw (-4,0) -- (4,0) node[pos=0.55, above] {};
      \draw (4,3) -- (-2,-3);
      \draw (2,3) -- (-4,-3);
       \draw (4,0) circle (2.1);
       \draw (-4,0) circle (2.1);
       \draw (1.9,3) -- (1.9, -3);
       \draw (-1.9,3) -- (-1.9, -3);
      \begin{scope}[black]
      \end{scope}[black]
       \filldraw (-4,0) circle (2pt);
       \filldraw (4,0) circle (2pt);
		\draw (-6.1,0) node[left] {$K$};
		\draw (6.1,0) node[right] {$K'$};

		
	\end{tikzpicture}
\begin{center}
\caption{Depiction of two distinct pairs of supporting hyperplanes, one being optimal.} \label{Fig3X}
\end{center}
\end{figure}

\begin{definition} \label{def5} Suppose $\delta_*>0$. We say a witness pair $(p,p') \in K \times K'$ is an $\epsilon$-{\it approximate solution to the supporting hyperplanes problem} if
$$d(p,p') - \delta_* \leq \epsilon d(p,p'),$$
and there exists a pair of parallel supporting hyperplanes $(H, H')$  orthogonal to the line segment $pp'$ such that the distance between them satisfies
$$\delta_* - d(H,H') \leq \epsilon d(p,p').$$
\end{definition}

In the above example of two discs, the best pair of supporting hyperplanes is the pair of lines that are orthogonal to the line connecting the centers.

Triangle Algorithm I computes  an $\epsilon$-approximate solution to the intersection problem when $\delta_*=0$, or a pair of separating hyperplane when $\delta_* >0$.  To describe the iterative step of the algorithm we need to give a definition.

\begin{definition}  \label{def6} Given a pair  $(p,p') \in K \times K'$ (see Figure \ref{Fig2}), we say $v \in K$ is a $p'$-{\it pivot} for $p$ if
\begin{equation} \label{def21}
d(p,v) \geq d(p',v).
\end{equation}
We say $v' \in K'$ is a $p$-{\it pivot} for $p'$ if
\begin{equation} \label{def22}
d(p',v') \geq d(p,v').
\end{equation}
\end{definition}

\begin{figure}[htpb]
	\centering
	\begin{tikzpicture}[scale=0.4]

	
		\draw (0.0,0.0) -- (7.0,0.0) -- (-2.0,-4.0) -- cycle;
		\draw (0,0) node[left] {$p'$};
		\draw (7,0) node[right] {$v$};
		\draw (-2,-4) node[below] {$p$};
           \filldraw (0,0) circle (2pt);
\filldraw (7,0) circle (2pt);
\filldraw (-2,-4) circle (2pt);
		

\draw (15.0,0.0) -- (22.0,0.0) -- (13.0,-4.0) -- cycle;
		\draw (15,0) node[left] {$p$};
		\draw (22,0) node[right] {$v'$};
		\draw (13,-4) node[below] {$p'$};
           \filldraw (15,0) circle (2pt);
\filldraw (22,0) circle (2pt);
\filldraw (13,-4) circle (2pt);

	\end{tikzpicture}
\begin{center}
\caption{$v$ is $p'$-pivot for $p$ (left); $v'$ is $p$-pivot for $p'$ (right).} \label{Fig2}
\end{center}
\end{figure}
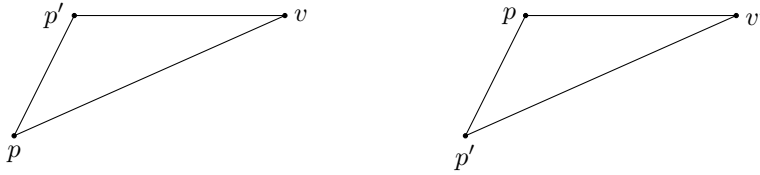

Consider the Voronoi diagram of the set $\{p, p' \}$ and the corresponding  Voronoi cells
\begin{equation}
V(p)= \{x: d(x,p) < d(x,p')\}, \quad  V(p')= \{x: d(x,p') < d(x,p)\}.
\end{equation}
If $H=\{x: h^Tx= a \}$ is the orthogonal bisecting hyperplane of the line $pp'$, it intersects  $K$ if and only if there exists $v \in K$ that is a $p'$-pivot for $p$, and $H$ intersects  $K'$ if and only if there exists
$v' \in K'$ that is a $p$-pivot for $p'$. In Figure \ref{Fig3}, the point $v$ and $v'$ are pivots for $p'$ and $p$, respectively. The four points $p,p',v,v'$ need not be coplanar.

\begin{figure}[htpb]
	\centering
	\begin{tikzpicture}[scale=0.6]	
      \draw (-4,0) -- (4,0) node[pos=0.55, above] {$H$};
      \draw (0,0) -- (4,0);
      \draw (-4,0) -- (-1,0);
      \begin{scope}[black]
      \end{scope}[black]
      \draw (0,-5) -- (0,1) node[pos=0.5, right] {};
       \filldraw (-4,0) circle (2pt);
       \filldraw (4,0) circle (2pt);
       \filldraw (-1,-5) circle (2pt);
		\draw (-4,0) node[left] {$p$};
		\draw (4,0) node[right] {$p'$};
        \draw (-1,-5) node[left] {$v'$};
        \draw (2,-3) node[left] {$v$};
        \filldraw (2,-3) circle (2pt);

           \filldraw (0,0) circle (2pt);
		
	\end{tikzpicture}
\begin{center}
\caption{Existence of pivot when $H$ intersects $K$ or $K'$.} \label{Fig3}
\end{center}
\end{figure}
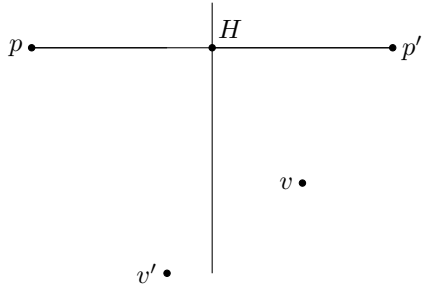

Each iteration of Triangle Algorithm I requires computing for a given pair $(p,p') \in K \times K'$ a $p'$-pivot $v$ for $p$, or a $p$-pivot $v'$ for $p'$. By squaring (\ref{def21}) and (\ref{def22}), these are respectively equivalent to checking if

\begin{equation} \label{pivots}
2v^T(p'-p) \geq \Vert p' \Vert^2 -  \Vert p \Vert^2, \quad 2v'^T(p-p') \geq \Vert p \Vert^2 -  \Vert p' \Vert^2.
\end{equation}

From the above it follows that the existence and computation of a pivot can be carried out by solving the convex programs that consist of optimization of a linear function over $K$ or $K'$. Specifically,

\begin{equation} \label{eqa5}
\max \{(p'-p)^Tv:  \quad v \in K\}, \quad \max \{(p-p')^Tv': \quad  v' \in K' \}.
\end{equation}

Let $T_K, T_{K'}$ denote the corresponding arithmetic complexities needed to solve the problems. Then the worst-case number of arithmetic operations in each iteration of Triangle Algorithm I is
\begin{equation} \label{eqa55}
T= \max \{T_K,T_{K'}\}.
\end{equation}
We prove that when $\delta_*=0$, the total number of required iterations to compute an $\epsilon$-approximate solution to the intersection problem is
\begin{equation}
O\bigg (\frac{1}{\epsilon^2} \bigg ).
\end{equation}

Consider

\begin{equation}
\Delta_0=\max\{d(x,y): x,y \in K\}, \quad \Delta'_0=\max\{d(x',y'): x',y' \in K'\},
\end{equation}
the diameters of $K$ and $K'$, respectively. Let
\begin{equation} \label{eqa3}
\rho_*= \max \{\Delta_0, \Delta_0'\}.
\end{equation}
When $\delta_* >0$, we prove the number of iterations of Triangle Algorithm I to compute a witness pair $(p,p') \in K \times K'$ is
\begin{equation} \label{eqa9}
O \bigg  ( \frac{\rho_*^2}{\delta_*^2} \bigg ).
\end{equation}
If one of the sets, say $K'$, is a single point, then any witness pair $(p,p')$ gives rise to an approximation to $\delta_*$ to within a factor of two:
\begin{equation} \label{eqhalf}
\frac{1}{2}d(p,p') \leq \delta_* \leq d(p,p').
\end{equation}

Triangle Algorithm II begins with a witness pair $(p, p') \in K \times K'$, then it computes an $\epsilon$-approximate solution to the distance problem.  Since $(p, p')$ is a witness pair there exists no $p'$-pivot for $p$, or a $p$-pivot for $p'$. However, if  $(p, p')$ is not
already an $\epsilon$-approximate solution to $\delta_*$, the algorithm makes use of a {\it weak-pivot}, defined next.

\begin{definition}  \label{def7} Given a witness pair  $(p,p') \in K \times K'$,  suppose that $H$ is the orthogonal bisecting hyperplane of the line segment $pp'$. We shall say
$v \in K$ is a {\it weak $p'$-pivot} for $p$ if it is not a $p'$-pivot but satisfies
\begin{equation}
d(p, H) >  d(v, H)
\end{equation}
(i.e. if $H_v$ is the hyperplane parallel to $H$ passing through $v$, it separates $p$ from $p'$, see Figure \ref{Fig7A}).  Similarly, we shall say
$v' \in K'$ is a {\it weak $p$-pivot} for $p'$ if  it is not a $p$-pivot but satisfies
\begin{equation}
d(p', H) >  d(v', H).
\end{equation}
\end{definition}

\begin{figure}[htpb]
	\centering
	\begin{tikzpicture}[scale=0.6]	
      \draw (-4,0) -- (4,0) node[pos=0.55, above] {$H$};
      \draw (0,0) -- (4,0);
      \draw (-4,0) -- (-1,0);
      \begin{scope}[black]
      \end{scope}[black]
      \draw (0,-6) -- (0,1) node[pos=0.5, right] {};
      \draw (-1,-6) -- (-1,1);
       \filldraw (-4,0) circle (2pt);
       \filldraw (4,0) circle (2pt);
       \filldraw (-1,-6) circle (2pt);
		\draw (-4,0) node[left] {$p$};
		\draw (4,0) node[right] {$p'$};
        \draw (-1,-6) node[left] {$v$};
         \filldraw (-1,0) circle (2pt);

           \filldraw (0,0) circle (2pt);
           \draw (-1,0) -- (-1,-6) node[pos=0.5, left] {$H_v$};
		
	\end{tikzpicture}
\begin{center}
\caption{The point $v$ is a weak $p'$-pivot for $p$, but not a $p'$-pivot.} \label{Fig7A}
\end{center}
\end{figure}
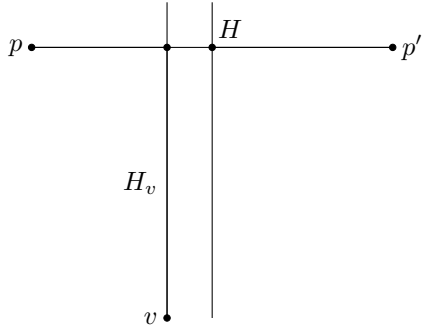

In an iteration of Triangle Algorithm II a given pair $(p_k,p'_k) \in K \times K'$ may or many not be a witness pair. The algorithm  searches for a weak-pivot or a pivot in order to reduce the current gap $\delta_k=d(p_k,p'_k)$ until $\epsilon$-approximate solutions to both the distance and supporting hyperplanes problems are reached.
Each iteration has complexity at most $T$,  see (\ref{eqa55}). We prove that the total number of iterations of Triangle Algorithm II is
\begin{equation} \label{eqb3}
O \bigg (\frac{\rho_*^2} {\delta_*^{2} }\frac{1}{\epsilon^{2}} \ln \frac {\rho_*}{\delta_*} \bigg ).
\end{equation}

To summarize, the total number of arithmetic operations in Triangle Algorithm I to compute an $\epsilon$-approximate solution to the intersection problem when $\delta_*=0$, the total number of arithmetic operations to compute a witness pair when $\delta_*>0$, and the
total number of arithmetic operations in Triangle Algorithm II to get both an $\epsilon$-approximate solution to the distance problem as well as  $\epsilon$-approximate solution to the supporting hyperplanes problem are, respectively

\begin{equation}
O\bigg ( T \frac{1}{\epsilon^2} \bigg), \quad
O \bigg  ( T\frac{\rho_*^2}{\delta_*^2} \bigg ), \quad
O \bigg ( T \frac{\rho_*^2} {\delta_*^{2} } \frac{1}{\epsilon^{2}} \ln \frac {\rho_*}{\delta_*}  \bigg ).
\end{equation}

Table 1 summarizes the complexity of Triangle Algorithms I and II in solving the general cases as well as several special cases described above.

\begin{table}[!t]
	\renewcommand{\arraystretch}{1.0}
	\centering
\scalebox{0.92
}{
\begin{tabular}{|l|l|l|c|}

\hline
Complexity  of  computing &  ~~~~~~~Intersection & ~~~~~~~Separation & Distance and Support
\\
 $\epsilon$-approximation solution& $~~~~~~K \cap K' \not = \emptyset $ & $~~~~~~K \cap K'= \emptyset $ &  $\delta_*=d(K,K')$
\\
	\hline
$(p,p') \in K \times K' \subset \mathbb{R}^m \times\mathbb{R}^m$ & & &
\\
$K, K'$  compact and convex& $d(p,p') \leq \epsilon d(p,v)$, or & $~~~~~~~~~(p,p')$ &
$d(p,p') - \delta_* \leq \epsilon d(p,p')$
\\
&  $d(p,p') \leq \epsilon d(p',v')$
& ~~~~~a witness pair & $\delta_* - d(H,H') \leq \epsilon d(p,p')$
\\
such that: & &  & $(H,H')$  supporting hyperplanes
\\
	\hline
$K = conv(\{v_1, \dots, v_n\})$  &  & &
\\
& $~~~~~~~O\big ( mn\frac{1}{\epsilon^2} \big)$ & $~~~~~~O \big  (mn \big (\frac{\rho_*}{\delta_*} \big)^2\big)$ & $~O \big (mn \big (\frac{\rho_*} {\delta_*\epsilon} \big)^2 \big )$
\\
complexity w. preprocessing  & $~~~~O\big ((m+n)\frac{1}{\epsilon^2} \big)$ & $~~O \big  ((m+n) \big (\frac{\rho_*}{\delta_*} \big)^2\big)$ & $~O \big ((m+n) \big (\frac{\rho_*} {\delta_*\epsilon} \big)^2 \big )$
\\
$K' = \{p'\}$  &  & &
\\
\hline
$K = conv(\{v_1, \dots, v_n\})$ & & &
\\

$N = \max \{n,n'\}$& $~~~~~~~O\big (m N \frac{1}{\epsilon^2} \big)$ & $~~~~~O \big (m N\big (\frac{\rho_*}{\delta_*} \big )^2 \big )$ & $O \big ( m N \big (\frac{\rho_*} {\delta_*\epsilon} \big)^2 \ln \frac {\rho_*}{\delta_*}  \big )$ 	
\\
complexity w. preprocessing & $~~~~O\big ((m +N )\frac{1}{\epsilon^2} \big)$ & $~~O \big ((m + N)\big (\frac{\rho_*}{\delta_*} \big )^2 \big )$ & $O \big ( (m + N) \big (\frac{\rho_*} {\delta_*\epsilon} \big)^2 \ln \frac {\rho_*}{\delta_*}  \big )$ 	
\\
$K' = conv(\{v'_1, \dots, v'_{n'}\})$ & & &
\\
\hline

$K=\{x: Ax \leq b \}$ & & &
\\
 $A$ an $n \times m$ matrix & $~~~~~~~O\big (mn \big)$ & $~~~~~~O \big  (T\big (\frac{\rho_*}{\delta_*} \big )^2 \big )$&$O \big (T \big (\frac{\rho_*} {\delta_*\epsilon} \big )^2 \big)$
\\
 $K'=\{p'\}$  &  & &
\\
\hline
$K=\{x: Ax \leq b\}$ &  & &
\\
 & $~~~~~~~O\big (T \frac{1}{\epsilon^2} \big)$ & $~~~~~~O \big  (T \big (\frac{\rho_*}{\delta_*} \big )^2 \big )$ & $O \big (T \big (\frac{\rho_*} {\delta_*\epsilon} \big )^2 \ln \frac {\rho_*}{\delta_*}  \big)$
\\
$K'=\{x: A'x \leq b' \} $& &  &
	\\
\hline

$K$ {\rm ~general} & & &
\\
& $~~~~~~~O\big (T\frac{1}{\epsilon^2} \big)$ & $~~~~~~O \big  (T \big (\frac{\rho_*}{\delta_*} \big )^2 \big )$ & $O \big ( T \big (\frac{\rho_*} {\delta_*\epsilon} \big )^2 \big)$
\\
$K'=\{p'\}$ & & &
\\
\hline

$K$ {\rm ~general} & & &
\\
& $~~~~~~~O\big (T\frac{1}{\epsilon^2} \big)$ & $~~~~~~O \big  (T \big (\frac{\rho_*}{\delta_*} \big )^2 \big )$ & $O \big ( T \big (\frac{\rho_*} {\delta_*\epsilon} \big )^2 \ln \frac {\rho_*}{\delta_*}  \big)$
\\
$K'$ {\rm ~general} & & &
\\
\hline

	\end{tabular}
}
\caption{The complexities of Triangle Algorithms I and II.   $T$ is the maximum of $T_K$ and $T_{K'}$, the complexities in optimizing a linear objective over $K$ and  $K'$. $\rho_*$ is maximum of diameters of $K$ and $K'$.}
\end{table}

\section{A New Separating Hyperplane Theorem}

In this section we first describe a new separation theorem. The theorem inspires an algorithmic separating hyperplane theorem that either computes an approximation to a point in the intersection of two compact convex sets, or a separating hyperplane. The algorithm can also approximate the distance between them when they are disjoint, as well as compute a pair of parallel supporting hyperplanes that approximates the optimal pair. First we give a well known definition.

\begin{definition} \label{def8}
Let $K$ be a compact convex subset in $\mathbb{R} ^m$. A point $v \in K$ is an extreme point of $K$ if it cannot be written as the convex combination of two distinct points in $K$. The set of all extreme points of $K$ is denoted by
${\rm ex}(K)$.
\end{definition}

The following finite dimensional version of Krein-Milman theorem is easily provable, see e.g. \cite{Barv}.


\begin{thm}  \label{thm1} {\rm(Krein-Milman)}  Let $K$ be a compact convex subset of $\mathbb{R} ^m$. Then $K$ is the convex hull of its extreme points. In notation, $K=conv({\rm ex}(K))$. $\Box$
\end{thm}

We will make use of it to prove the following.

\begin{thm}  \label{thm2} {\rm (Distance Duality)}  Let $K, K'$ be compact convex subsets of $\mathbb{R} ^m$, with  ${\rm ex}(K)$ and ${\rm ex}(K')$ as their corresponding set of extreme points. Let $S$ be a subset of $K$ containing ${\rm ex}(K)$, and $S'$ a subset of $K'$ containing ${\rm ex}(K')$. Then, $K \cap K' \not = \emptyset$ if and only if for each $(p, p')  \in K \times K'$, either
there exists  $v \in S$ such that $d(p, v) \geq d (p', v)$, or there exists $v' \in S'$ such that $d(p', v') \geq d(p, v')$.
\end{thm}

\begin{proof} Suppose $K \cap K' \not = \emptyset$.  Let $(p,p')  \in K \times K'$ be given. If $p=p'$, the result is obvious. So assume $p \not = p'$. Consider the Voronoi diagram of the two point set $\{p, p' \}$ and the corresponding  Voronoi cells

\begin{equation}
V(p)= \{x: d(x,p) < d(x,p')\}, \quad V(p')= \{x: d(x,p') < d(x,p)\}.
\end{equation}

Suppose there does not  exist $v \in S$ such that $d(p, v) \geq d(p',v)$.  Then $S \subset V(p)=\{x: d(x,p) < d(x,p')\}$. However, since $V(p)$ is convex, we must have $conv({\rm ex}(K)) \subset V(p)$.  But by Krein-Milman Theorem, $K=conv({\rm ex}(K))$. Thus $K \subset V(p)$.
Suppose also there does not  exist $v' \in S'$ such that $d(p', v') \geq d(p,v')$.  Then by an analogous argument $K' \subset V(p')$. But $V(p)$ and $V(p')$ are disjoint, contradicting that $K$ intersects $K'$.

Conversely, suppose $K \cap K' = \emptyset$.  Then $\min \{d(x,x'): x \in K, x' \in K'\}>0$, is attained at some $(p_*, p'_*) \in K \times K'$.  We claim the orthogonal bisecting hyperplane $H$ of $p_*p'_*$ separates $K$ and $K'$. Suppose $H$ intersects $K$ at a point $q$ (see Figure \ref{Fig88} for a 2D depiction). Then by convexity, the entire line segment $p_*q$ lies in $K$. Considering the isosceles triangle $\triangle p_*qp'_*$,  we can argue there exists a point $u$ on $p_*q$, see Figure \ref{Fig88}, which is closer to $p'_*$. This is a contradiction.
\end{proof}

\begin{figure}[htpb]
	\centering
	\begin{tikzpicture}[scale=0.6]

      \draw (-4,0) -- (2,0) node[pos=0.55, above] {$H$};
      \draw (0,0) -- (2,0) node[pos=0.5, above] {};
      \draw (-4,0) -- (-1,0) node[pos=0.5, above] {};
      \draw (-1,0) -- (0,0) node[pos=0.55, above] {};
      \draw (2,0) -- (-2.8,-2.4) node[pos=0.55, below] {};
      \draw (-4,0) -- (-2.8,-2.4);
      \filldraw (-2.8,-2.4) circle (2pt);
      \draw (-2.8,-2.4) node[left] {$u$};
      \draw (-1,-6) -- (-1,1) node[pos=0.3, left] {};
       \filldraw (-4,0) circle (2pt);
       \filldraw (2,0) circle (2pt);
       \filldraw (-1,-6) circle (2pt);
		\draw (-4,0) node[left] {$p_*$};
		\draw (2,0) node[right] {$p_*'$};
        \draw (-1,-6) node[left] {$q$};
         \draw (-1,0) node[above] {};
         \filldraw (-1,0) circle (2pt);
         \draw (-4,0) -- (-1,-6) node[pos=0.3, left] {};

           \draw (-1,0) -- (-1,-6) node[pos=0.5] {};
		
	\end{tikzpicture}
\begin{center}
\caption{A 2D depiction of the case where the orthogonal bisecting hyperplane intersects $K$.} \label{Fig88}
\end{center}
\end{figure}
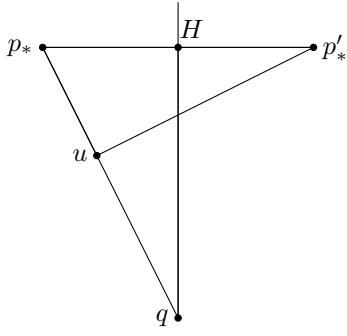

An alternative description of the distance duality, Theorem \ref{thm2}, is the following version:

\begin{thm}  \label{thm3} {\rm (Distance Duality)} Let $K, K'$ be compact convex subsets in $\mathbb{R} ^m$, with  ${\rm ex}(K)$ and ${\rm ex}(K')$ as their corresponding set of extreme points.
Then, $K \cap K'= \emptyset$ if and only if there exists $p \in K$, $p' \in K'$ such that $d(p, v) < d(p', v)$ for all $v \in {\rm ex}(K)$ and $d(p', v') < d(p, v')$ for all $v' \in {\rm ex}(K')$.
\end{thm}

\begin{proof} Suppose $K \cap K' = \emptyset$.  Let $d(K,K')=d(p_*, p'_*)$, $(p_*, p'_*) \in K \times K'$. Then as shown in Theorem \ref{thm2}, the orthogonal bisector of $p_*p'_*$ separates $K$ and $K'$. This implies the strict inequalities.

Conversely, suppose there exists $(p, p') \in K \times K'$ satisfying the strict inequalities for all the extreme points.  Then the orthogonal bisector hyperplane of $pp'$ separates ${\rm ex}(K)$ and ${\rm ex}(K')$.  By convexity of $K, K'$, this hyperplane must separate the convex hulls of ${\rm ex}(K)$ and ${\rm ex}(K')$. Then by the Krein-Milman Theorem $K \cap K' = \emptyset$.
\end{proof}

\begin{remark} \label{rem1}
We can view a pair $(p_*,p'_*) \in K \times K'$ such that $d(K,K') =d(p_*, p'_*)$ as a special witness pair.
\end{remark}

When $K$ is a finite point set and $K'$ a singleton (the convex hull membership problem), the distance dualities reduce to the characterization theorems in \cite{kal14}.

\begin{prop} Suppose $d(K,K')=d(p_*, p'_*)$, where $(p_*, p'_*) \in K \times K$.  Then  if $H_{p_*}$  and $H_{p'_*}$  are orthogonal hyperplanes  to the line segment $p_*p'_*$ at $p_*$ and $p'_*$ respectively, they are optimal supporting hyperplanes to $K$ and $K'$, respectively. In other words,  $d(K,K')=d(p_*, p'_*)=d(H_{p_*}, H_{p'_*})$.
\end{prop}

\begin{proof} Assume one of these hyperplanes is not a supporting hyperplane, say $H_{p_*}$. Then it must intersect $K$ at another point $v$ lying strictly between $H_{p_*}$ and $H$. But then by convexity of $K$,  the line segment $p_*v$ lies in $K$, see Figure \ref{Fig89}. We can thus choose a point $w$ on the line segment $p_*v$ so that in the triangle $\triangle p_*wp'_*$ the largest side is $d(p_*, p'_*)$. This contradicts that $d(p_*,p'_*)=d(K,K')$.
\end{proof}

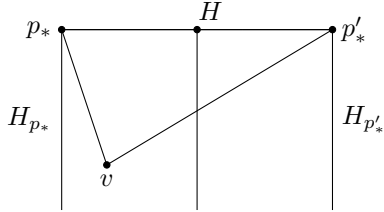
\begin{figure}[htpb]
	\centering
	\begin{tikzpicture}[scale=0.6]

      \draw (-4,0) -- (2,0) node[pos=0.55, above] {$H$};
      \draw (0,0) -- (2,0) node[pos=0.5, above] {};
      \draw (-4,0) -- (-1,0) node[pos=0.5, above] {};
      \draw (-1,0) -- (0,0) node[pos=0.55, above] {};
      \filldraw (-3,-3) circle (2pt);
      \draw (-3,-3) node[below] {$v$};
      \draw (-4,0) -- (-4,-4) node[pos=0.5, left] {$H_{p_*}$};
      \draw (2,0) -- (2,-4) node[pos=0.5, right] {$H_{p'_*}$};
       \filldraw (-4,0) circle (2pt);
       \filldraw (2,0) circle (2pt);
		\draw (-4,0) node[left] {$p_*$};
		\draw (2,0) node[right] {$p_*'$};
         \draw (-1,0) node[above] {};
         \filldraw (-1,0) circle (2pt);
          \draw (-3,-3.) -- (2,0) node[pos=0.3, left] {};
          \draw (-3,-3.) -- (-4,0) node[pos=0.8, left] {};

           \draw (-1,0) -- (-1,-4) node[pos=0.5] {};
		
	\end{tikzpicture}
\begin{center}
\caption{If $H_{p_*}$ is not supporting $K$, $d(p_*,p'_*)$ is not optimal.} \label{Fig89}
\end{center}
\end{figure}

\begin{definition} \label{def9} Given $p \in K$ and $p' \in K'$, we shall say $p$ is {\it witness to the infeasibility} of $p'$ in $K$ if $d(p, v) < d(p', v)$ for all $v \in {\rm ex}(K)$. Equivalently, if the orthogonal bisector of $pp'$ separates $p'$ from $K$. We denote the set of all such witnesses in $K$ as $W_{p'}(K)$. Analogously, we shall say $p'$ is a witness to the infeasibility of $p$ in $K'$ if $d(p', v') < d(p, v')$ for all $v' \in {\rm ex}(K')$. Equivalently, if the orthogonal bisector of $pp'$ separates $p$ from $K'$. We denote the set of all such witnesses in $K'$ as $W_p(K')$.
\end{definition}

 The orthogonal bisector of any witness pair separates $K$ and $K'$.  However, unlike the case when $K'$ is a singleton element, a witness pair does not estimate $d(K,K')$ to within a factor of two.  Also, if $p$ is a witness to the infeasibility of $p'$ in $K$, the orthogonal bisecting hyperplane of $pp'$ does not necessarily separate $K$ and $K'$.

The distance duality theorems stated above generalize the corresponding dualities for the case when  $K$ is the convex hull of a finite number of points and $K'$ is a singleton point. Additionally,  Triangle Algorithm I to be described here generalizes our earlier Triangle Algorithm in \cite{kal14}.  It either computes an $\epsilon$-approximation solution to the intersection problem when $d(K,K') =0$, or a witness pair.  Then Triangle Algorithm II takes over and computes an $\epsilon$-approximate solution to an optimal pair $(p_*,p'_*)$,
 as well as an  $\epsilon$-approximate solution to an optimal pair of supporting hyperplanes  $(H_{p_*}, H_{p'_*})$.

\section{A Theorem for Iterative Improvement of Distance}

Given any pair $(p,p') \in  K \times K'$, $d(p,p')$ provides an upper bound to $\delta_*=d(K,K')$, the distance between the two convex sets.  Our goal is to iteratively compute better estimates of $\delta_*$. Specifically, we will accomplish four tasks:

\begin{itemize}

\item

When $\delta_*=0$ (i.e. $K \cap K' \not = \emptyset$), we will compute $(p, p')  \in K \times K'$, so that $d(p,p')$ is to within a prescribed tolerance (see Definition \ref{def1}).

\item

When $\delta_* >0$ (i.e. $K \cap K'= \emptyset$) we will compute $(p, p')  \in K \times K'$ so that the orthogonal bisecting hyperplane  of the line segment $pp'$, say $H$,  separates $K$ and $K'$, i.e. $(p,p')$ is a witness pair (see Definition \ref{def2}).

\item

When $\delta_* >0$, we will compute $(p, p')  \in K \times K'$ so that it is a witness pair and $d(p,p')$ is within a prescribed tolerance  of $\delta_*$ (see Definition \ref{def3}).

\item

When $\delta_* >0$, we will compute a pair of supporting hyperplanes $(H, H')$, so that $d(H, H')$ is to within a prescribed tolerance  of $\delta_*$ (see Definition \ref{def4}).

\end{itemize}

We will next prove an error bound that will be used in the analysis of complexity of the iterative step in any of the above mentioned four tasks. The theorem to be proved is a general result that reveals a significant property of three points in the Euclidean plane.  It is a stronger and yet more convenient version of a theorem proved in \cite{kal14}.  The reader may assume $p \in K$, $p' \in K'$, and $v' \in K'$ is a $p$-pivot for $p'$ (see Definition \ref{def6}).  However, the theorem alternatively is valid for the case when $v \in K$ is a $p'$-pivot for $p$.

\begin{thm}  \label{thm4}  Let $p, p', v'$  be distinct points in $\mathbb{R} ^m$. Suppose $d(p',v') \geq d(p,v')$ {\rm (see Figure \ref{Fig3A})}.  Let $p''$ be the point on the line segment $p'v'$  that is closest to $p$.  Let $\delta=d(p',p)$, $\delta'=d(p'',p)$, and $r=d(p,v')$. Let $C$ be the circle of radius $r$ centered at $p$
and $C'$ the circle of radius $r$ centered at $v'$.  Let $C''$ be the circle of radius $\delta$ centered at $p$. If $\delta >r$,  let the intersection of the line segment $p'v'$ with $C'$ be denoted by $\overline p'$ {\rm (see Figure \ref{Fig5}, $p_1'$, $p_2'$, $p_3'$)}. Let  $\overline \delta=d(\overline p',p)$. Let $\theta=\angle pv'p'$. Then, $\delta' \leq \overline \delta \leq \delta$, and

\begin{equation} \label{gap}
\delta' \leq
\begin{cases}
\delta \sqrt{1- \frac{\delta^2}{4 r^2}} \leq \delta \exp \big (-\frac{\delta^2}{8 r^2} \big ), &\text{if $\delta \leq r$;}\\
\\
\overline \delta \sqrt{1- \frac{\overline \delta^2}{4 r^2}} \leq \overline \delta \exp \big (-\frac{\overline \delta^2}{8 r^2} \big ) \leq
\delta \exp \big (-\frac{\overline \delta^2}{8 r^2} \big ), &\text{if $\delta > r$, $\overline \delta \leq r$, $0 \leq \theta \leq \frac{\pi}{3}$;}\\
\\
\frac{\sqrt{3}}{2} \overline \delta \leq \frac{\sqrt{3}}{2} \delta, &\text{if $\delta > r$, $\overline \delta > r$, $\frac{\pi}{3} < \theta < \frac{\pi}{2}$;}\\
\\
r \leq \frac{\sqrt{2}}{2} \delta, &\text{if $\delta > r$, $\overline \delta \leq r$, $\theta \geq \frac{\pi}{2}$.}
\end{cases}
\end{equation}
\end{thm}

\begin{proof} We will prove the four cases in (\ref{gap}) case by case.
Without loss of generality we may assume they lie in the Euclidean plane.

Case (i): $\delta \leq r$, see Figure \ref{Fig3A}. Consider $p'$ as a variable $x'$ and the corresponding $p''$ as $x''$. We will consider the maximum value of $d(x'',p)$ subject to the desired constraints. We will prove
\begin{equation}
\delta^* = \max \big \{d(x'',p):  \quad x \in \mathbb{R} ^2, \quad d(x',p) =\delta, \quad d(x',v') \geq r \big \}= \delta \sqrt{1- \frac{\delta^2}{4 r^2}}.
\end{equation}

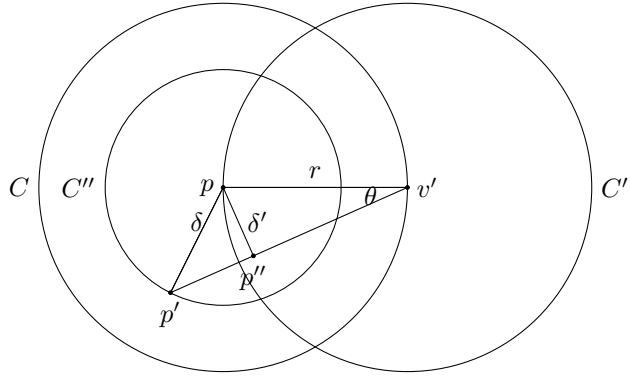
\begin{figure}[htpb]
	\centering
	\begin{tikzpicture}[scale=0.35]


\begin{scope}[black]
         \draw (0.0,0.0) circle (7.0);
		 \draw (7.0,0.0) circle (7.0);
		 \draw (0.0,0.0) circle (4.48);
\end{scope}
		
		\draw (0.0,0.0) -- (7.0,0.0) -- (-2.0,-4.0) -- cycle;
      \draw (0,0) -- (7,0) node[pos=0.5, above] {$r$};
      \draw (-2,-4) -- (0,0) node[pos=0.5, above] {$\delta$};
       \draw (0,0) -- (1.15,-2.6) node[pos=0.5, right] {$\delta'$};
       \draw (1.15,-2.6) node[below] {$p''$};
       \filldraw (1.15,-2.6) circle (2pt);
		\draw (0,0) node[left] {$p$};
		\draw (7,0) node[right] {$v'$};
	\draw (5,-.3) node[right] {$\theta$};
		\draw (-2,-4) node[below] {$p'$};
         \draw (14,0) node[right]{$C'$};
          \draw (-7,0) node[left]{$C$};
           \draw (-4.48,0) node[left]{$C''$};
           \filldraw (0,0) circle (2pt);
\filldraw (7,0) circle (2pt);
\filldraw (-2,-4) circle (2pt);
		
	\end{tikzpicture}
\begin{center}
\caption{Depiction of gaps $\delta=d(p',p)$, $\delta'=d(p'',p)$, when $\delta \leq r=d(p,v')$.} \label{Fig3A}
\end{center}
\end{figure}

Given that $\delta \leq r$, $p'$ must lie inside or on the boundary of the circle of radius $\delta$ centered at $p$ (i.e. $C$), but outside or on the boundary of the circle of radius $r$ centered at $v'$ (i.e. $C'$),
see Figure \ref{Fig3A}.

Consider the ratio $\delta'/r$ as $p'$ ranges over all the points on the circumference of $C''$ while
outside or on the boundary of $C'$. It is  geometrically obvious and easy to argue that this ratio is maximized when $p'$ is a point of intersection of the circles $C'$ and $C''$, denoted by $p'_*$ in Figure \ref{Fig4}. We now compute the corresponding ratio.

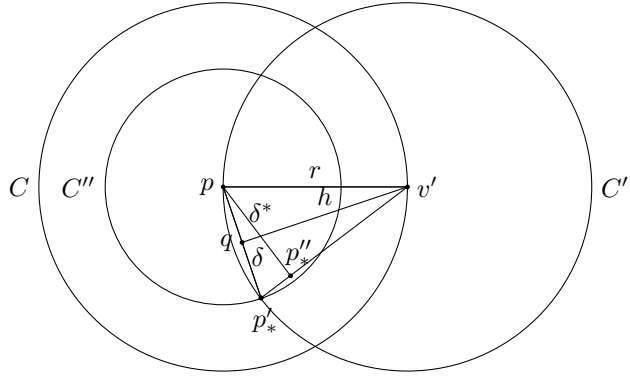
\begin{figure}[htpb]
	\centering
	\begin{tikzpicture}[scale=0.35]


\begin{scope}[black]
         \draw (0.0,0.0) circle (7.0);
		 \draw (7.0,0.0) circle (7.0);
		 \draw (0.0,0.0) circle (4.48);
\end{scope}
		
		\draw (0.0,0.0) -- (7.0,0.0) -- (1.44,-4.22) -- cycle;
      \draw (0,0) -- (7,0) node[pos=0.5, above] {$r$};
       \filldraw (0,0) circle (2pt);
       \filldraw (7,0) circle (2pt);
       \draw (.72,-2.11) -- (7,0) node[pos=0.5, above] {$h$};
       \draw (2.56,-3.38) -- (0,0) node[pos=0.5, above] {$~\delta^*$};
       \filldraw (0.72,-2.11) circle (2pt);
\filldraw (2.56,-3.38) circle (2pt);
       \draw (1.44,-4.22) -- (0,0) node[pos=0.2, above] {$~\delta$};
       \draw (2.56,-3.38) node[above] {$~~p''_{*}$};
       \filldraw (1.44,-4.22) circle (2pt);
\filldraw (1.44,-4.22) circle (2pt);
       \draw (1.44,-4.22) node[below] {$~p'_*$};
        \draw (.72,-2.11) node[left] {$q$};
		\draw (0,0) node[left] {$p$};
		\draw (7,0) node[right] {$v'$};
         \draw (14,0) node[right]{$C'$};
          \draw (-7,0) node[left]{$C$};
           \draw (-4.48,0) node[left]{$C''$};
		
	\end{tikzpicture}
	
\begin{center}
\caption{The worst-case scenario for the gap $\delta'= \delta^*$, when $\delta \leq r$.} \label{Fig4}
\end{center}
\end{figure}

Consider Figure \ref{Fig4}, and the isosceles triangle  $\triangle v'pp'_*$. Let $q$ denote the midpoint of $p$ and $p'_*$. Let $h$ denote $d(q,v')$. Let $p''_*$ be the nearest point to $p$ on the line segment $p'_*v'$.  Consider the right triangles $\triangle pv'q$  and  $\triangle pp'_*p''_{*}$. The angles  $\angle v'pq$  and  $\angle pp'_*p''_{*}$ are equal. Hence, the two triangles are similar and we may write

\begin{equation}
\frac{\delta^*}{\delta} = \frac{h}{r}=\frac{1}{r}\sqrt{{r^2} - \frac{\delta^2}{4}}= \sqrt{1- \frac{\delta^2}{4 r^2}}.
\end{equation}
This proves the first inequality in the first case of (\ref{gap}). To prove the next inequality for this case, we use the fact that for any real $t$,  $1+t \leq \exp(t)$, and set $t= -{\delta^2}/{4 r^2}$.

Next we assume $\delta >r$ and consider the three remaining cases of the theorem according the values of $\theta=\angle pv'p'$: $0 \leq \theta \leq \pi/3$,
$ \pi/3  < \theta < \pi/2$, and $\theta \geq \pi/2$.  Figure \ref{Fig5} considers one example of each possible case, corresponding to $p_1'$, $p_2'$, and $p_3'$, respectively. Note that
$$\delta' \leq \overline \delta \leq \delta.$$
This is straightforward by considering the triangle $\triangle pp''p'$ (it is a right triangle when $\theta$ is acute).

\begin{figure}[htpb]
	\centering
	\begin{tikzpicture}[scale=0.45]
\begin{scope}[black]
         \draw (0.0,0.0) circle (7.0);
		 \draw (7.0,0.0) circle (7.0);
\end{scope}
		\draw (0.0,0.0) -- (7.0,0.0) -- (-7, -10) -- cycle;
\draw (0.0,0.0) -- (7.0,0.0) -- (10, -11) -- cycle;
\draw (0.0,0.0) -- (7.0,0.0) -- (3.8, -8) -- cycle;
\draw (-7,-10) node[below] {$p_1'$};
\filldraw (1.26,-4.1) circle (2pt);
\draw (1.26,-4.1) node[left] {$\overline p_1'$};
\draw (0,0) -- (1.26,-4.1) node[pos=0.5,left] {$\overline \delta$};
\draw (1.26,-4.1) --(0,0);
\draw (3.8,-8.0) node[below] {$p_2'$};
\filldraw (3.8,-8.0) circle (2pt);
\draw (0,0) -- (3.8,-8.0) node[pos=0.5, above] {};
\draw (10,-11) node[below] {$p_3'$};
\filldraw (4.33,-6.5) circle (2pt);
\draw (4.33,-6.5) node[right] {$\overline p_2'$};
\filldraw (-7,-10) circle (2pt);
\filldraw (10,-11) circle (2pt);
\filldraw (0,0) circle (2pt);
\filldraw (7,0) circle (2pt);
\draw (-7,-10) -- (0,0) node[pos=0.5, left] {$\delta$};
\draw (10,-11) -- (0,0) node[pos=0.3, left] {$\delta$};
\draw (3.8,-8.0) -- (0,0) node[pos=0.2, left] {$\delta$};
\draw (10,-11) -- (7,0) node[pos=0.3, right] {$\delta$};
      \draw (0,0) -- (7,0) node[pos=0.5, above] {$r$};
       \draw (2.3,-3.38) node[below] {$~~p_1''$};
       \draw (2.3,-3.38) --(0,0) node[pos=0.2] {$\delta'$};
       \filldraw (2.3,-3.38) circle (2pt);
        \filldraw (6,-2.5) circle (2pt);
        \draw (6,-2.5) node[below] {$~~p_2''$};
        \draw (6,-2.5) -- (0,0) node[pos=0.3, left] {$\delta'$};
		\draw (0,0) node[left] {$p$};
		\draw (7,0) node[right] {$v'=p_3''$};
         \draw (14,0) node[right]{$C'$};
          \draw (-7,0) node[left]{$C$};

	\end{tikzpicture}
\begin{center}
\caption{Depiction of gaps $\delta=d(p',p)$, $\delta'=d(p'',p)$, when $\delta > r=d(p,v)$.} \label{Fig5}
\end{center}
\end{figure}
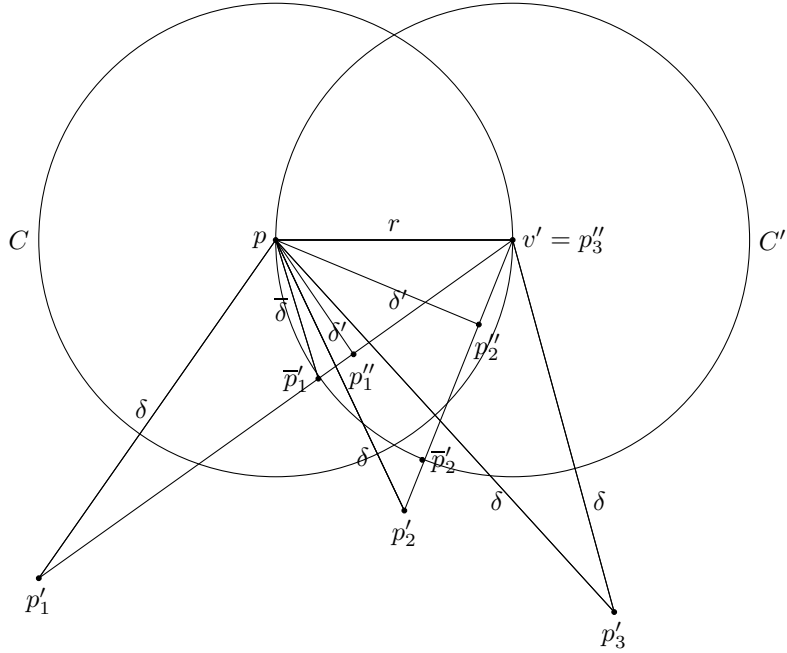

Case (ii): $\delta > r$,  $0 \leq \theta \leq \pi/3$.  Since $\theta$ is acute $\overline p'$ lies inside of $C$. Replacing $p'$ by $\overline p'$, $p''$ remains unchanged (see $p_1'$, $\overline p_1'$, and $p_1''$ in Figure \ref{Fig5}).   In  this case $\overline \delta \leq r$, then we are back to the first case and the same analysis applies with $\delta$ replaced with $\overline \delta$. This together with the inequalities $1+t \leq \exp(t)$ gives the proof of the two inequalities in the second case of (\ref{gap}).

Case (iii): $\delta > r$,  $\pi/3 < \theta < \pi/2$ (see $p_2'$, $\overline p_2'$, and $p_2''$ in Figure \ref{Fig5}). In this case too it can be shown that
$$\delta'=\overline \delta \sqrt{1- \frac{\overline \delta^2}{4 r^2}}.$$
We determine when the right-hand-side of above quantity is maximized, given that $\theta$ lies in the above range.  It is easy to show the maximum occurs for $\theta=\pi/3$, corresponding to the case  where $\overline p'$ lies on the intersection of  $C$ and $C'$.  Equivalently,  $\theta= \pi/3$ gives $\overline \delta =\sqrt{3}r/2$. This gives the first claimed inequality in case 3 of (\ref{gap}). Next, the fact that
$\overline \delta \leq \delta$ proves the next inequality in this case.

Case (iv): $\delta > r$,  $\pi/2 \leq
 \theta$. In this case $p''$ coincides with $v$ (see $p''_3$ in Figure \ref{Fig5}). Trivially we have,  $\delta' \leq r \leq  \delta \sqrt{2}/2$.
\end{proof}

\section{Algorithm for Testing Intersection or Separation of Convex Sets}  \label{sec4}

In this section and next we describe a simple algorithm, referred as {\it Triangle Algorithm I}. This is a generalization of the original Triangle Algorithm for the special case when $K$ is the convex hull of a finite set of points and $K'$ a singleton point.  In contrast with the earlier version, Triangle Algorithm I applies to the case where $K$ and $K'$ are arbitrary compact convex sets.
The justification in the name of the algorithm lies in the fact that in each iteration, given a pair $(p,p') \in K \times K'$,  where $d(p,p')$ is not yet satisfactory, the algorithm searchers for a pivot, either in $K$ or in $K'$ so as to reduce the gap $d(p,p')$. Specifically, the algorithm searches for a triangle $\triangle pp'v'$ where $v'$ lies in a subset $S'$ of $K'$ containing ${\rm ex}(K')$ (extreme points of $K'$), $p' \in K'$, where $d(p',v') \geq d(p,v')$; or
a triangle $\triangle pp'v$ where $v$ lies in a subset $S$ of $K$ containing ${\rm ex}(K)$, $p \in K$, where $d(p,v) \geq d(p',v)$.  Given that such triangle exists, it uses $v$ or $v'$ as a pivot to bring $p,p'$ in current iterate $(p,p') = (p_k,p_k') \in K \times K'$  closer to each other by generating either a new iterate $p_{k+1} \in K$, or new iterate $p'_{k+1} \in K'$ such that if we denote the new iterate by $(p_{k+1}, p'_{k+1})$,
$d(p_{k+1}, p'_{k+1}) < d(p_k,p'_k)$.   Theorem \ref{thm4} assures a certain reduction in terms of $d(p_k,p'_k)$ itself. If no such a triangle exists, then by Theorem \ref{thm2}, $(p_k, p'_k)$ is a witness pair certifying that $K$ and $K'$ do not intersect.

\begin{definition} \label{def10} Given three points $x,y,z \in \mathbb{R}^m$ such that $d(y,z) \geq d(x,z)$. Let $nearest(x; yz)$ be the nearest point to $x$ on the line segment joining  $y$ to $z$.
\end{definition}

We have

\begin{prop} Given three points $x,y,z \in \mathbb{R}^m$, let the {\it step-size} be
\begin{equation}
\alpha = \frac{(x-y)^T(z-y)}{d^2(y,z)}.
\end{equation}
Then
\begin{equation} \label{pdp}
nearest(x; yz)=
\begin{cases}
(1-\alpha)y + \alpha z, &\text{if $\alpha \in [0,1]$;}\\
z, &\text{otherwise.} ~~~\Box
\end{cases}
\end{equation}
\end{prop}

\section{Triangle Algorithm I: Properties and  Complexity Analysis}

Here we describe Triangle Algorithm I for testing if two compact convex sets $K, K'$ intersect. It computes a pair $(p,p')  \in K \times K'$ such that either $d(p,p')$ is to within a prescribed tolerance, or it is a witness pair. It assumes we are given points $(p_0, p_0') \in K \times K'$ and $\epsilon \in (0,1)$.

\begin{center}
\begin{tikzpicture}
\node [mybox] (box){%
    \begin{minipage}{0.93\textwidth}
{\bf  Triangle Algorithm I ($(p_0, p'_0) \in K \times K'$, $\epsilon \in (0,1)$)}\

{\bf Step 0.} Set $p=v=p_0$, $p'=v'= p_0'$.

{\bf Step 1.} If $d(p,p') \leq \epsilon d(p,v)$, or $d(p,p') \leq \epsilon d(p',v')$, stop.

{\bf Step 2.}  Test if there exists $v \in K$ that is a $p$-pivot for $p'$, i.e.
\begin{equation} \label{pivotsx}
v^T(p'-p) \geq  \frac{1}{2} (\Vert p' \Vert^2 -  \Vert p \Vert^2)
\end{equation}
(e.g. set $v= {\rm argmax} \{(p'-p)^Tv:  v \in K\}$).  If  pivot exists,
set $p \leftarrow nearest(p'; pv)$. Go to Step 1.

{\bf Step 3.}  Test if there exists $v' \in K'$ that is a $p'$-pivot for $p$, i.e.
\begin{equation} \label{pivotsxx}
v'^T(p-p') \geq  \frac{1}{2}(\Vert p \Vert^2 -  \Vert p' \Vert^2)
\end{equation}
(e.g. set  $v'={\rm argmax} \{(p-p')^Tv':  v' \in K' \}$).  If pivot exists,
set $p' \leftarrow nearest(p; p'v')$. Go to Step 1.


{\bf Step 4.} Output $(p,p')$ as a witness pair, stop ($K \cap K' = \emptyset$).


    \end{minipage}};
\end{tikzpicture}
\end{center}

\begin{remark}  \label{remcomp} From the computational point of view it is important to note that in Step 2 of Triangle Algorithm I the search for a pivot does not necessarily require solving a linear programming to optimality. For instance, in the implementation of the Triangle Algorithm for the convex hull membership problem,  a pivot may be found by randomly checking a constant number of $v_i$'s. Thus in such cases the complexity of finding a pivot is merely $O(m)$. This together with updating the new iterate results in a complexity of $O(m+n)$ operations per iteration.
For some heuristic ideas and computational results, see \cite{Meng}. In many practical cases, the number of points in $V$ or $V'$ is much larger than $m$ and we may attempt to keep the iterates $(p,p')$ so that $p$  has a representation in terms of $O(m)$ of the $n$ points of $V$, and also $p'$  a representation in terms of $O(m)$ of the $n'$ points of $V'$ .  This means a typical complexity may even be further reduced to $O(m)$, as opposed to $O(m+\max \{n,n'\})$. In Section 8 we will discuss the complexity of these problem in more detail.
\end{remark}

\begin{lemma}  \label{lem1}  Assume $K \cap K' \not = \emptyset$. Let $\rho_*$ be the maximum of the diameters of $K$ and $K'$ (see \ref{eqa3}).  Assume $(p_0, p'_0)  \in K \times K'$ is given. Let $\delta_0=d(p_0,p'_0)$. Let $k\equiv k(\delta_0)$ be the maximum number of iterations of Triangle Algorithm I  to halve the error, i.e.
compute a pair $(p_k, p'_k) \in K \times K'$ so that if $\delta_j=d(p_j,p'_j)$ for $j=1, \dots, k$, we have
\begin{equation}
\delta_k \leq \frac{\delta_0}{2} < \delta_j,  \quad j=1, \dots, k-1.
\end{equation}
Then, $k$ satisfies
\begin{equation} \label{iter}
k= k(\delta_0) \leq  \lceil N_0 \rceil, \quad   N_0 \equiv N(\delta_0) =(32 \ln 2)  \frac{\rho_*^2}{\delta_0^2} <  23 \frac{\rho_*^2}{\delta_0^2}.
\end{equation}
\end{lemma}

\begin{proof}  From the description of Triangle Algorithm I, and the bounds in  (\ref{gap}) in Theorem \ref{thm4}, for each $j=1, \dots, k-1$, either
\begin{equation}
\delta_j \leq \delta_{j-1} \exp \bigg ({-\frac{\delta^2_{j-1}}{8 \rho_*^2}} \bigg ),
\end{equation}
or
\begin{equation}
\delta_j \leq \delta_{j-1} \exp \bigg ({-\frac{\overline \delta^2_{j-1}}{8 \rho_*^2}} \bigg ),
\end{equation}
or
\begin{equation}
\delta_j \leq \frac{\sqrt{3}}{2} \delta_{j-1}.
\end{equation}
We have used the first three cases of (\ref{gap}), and the  fact that in the forth case $\sqrt{2} < \sqrt{3}$. Since $d(p_j,p_j')$ is monotonically decreasing in $j$ and by assumption for each $j \leq k-1$ we have,
\begin{equation}
\frac{1}{2}\delta_0 \leq \delta_{j-2} \leq \overline \delta_{j-1}.
\end{equation}
It follows that for each $j \leq k-1$ we have, either
\begin{equation} \label{bound}
\delta_j \leq \delta_{j-1} \exp \bigg ({-\frac{\delta^2_{0}}{32 \rho_*^2}} \bigg ),
\end{equation}
or
\begin{equation} \label{x}
\delta_j \leq \frac{\sqrt{3}}{2} \delta_{j-1}.
\end{equation}

Thus from (\ref{bound}) and (\ref{x}) we may write
\begin{equation}
\delta_k \leq \delta_0 \bigg (\frac{\sqrt{3}}{2} \bigg )^{k_1}\exp \bigg ({-\frac{k_2\delta_0^2}{32 \rho_*^2}} \bigg ),
\end{equation}
where $k_1$ and $k_2$ are nonnegative integers satisfying
\begin{equation}
k_1+k_2=k.
\end{equation}
To have $\delta_k \leq \delta_0/2$, it suffices to satisfy

\begin{equation}
\bigg (\frac{\sqrt{3}}{2} \bigg )^{k_1}\exp \bigg ({-\frac{k_2\delta_0^2}{32 \rho_*^2}} \bigg ) \leq \frac{1}{2}.
\end{equation}
Clearly the worst-case is when $k_1=0$ and $k_2=k$. Then solving for $k$ in the above inequality implies

$$k \leq \lceil 32 \ln (2) \frac{\rho_*^2}{\delta_0^2}  \rceil.$$
\end{proof}

\begin{thm}  \label{thm5}  Triangle Algorithm I satisfies the following properties:

(i) Suppose $\delta_*=d(K,K')=0$. Given $\epsilon >0$, the number of iterations $k_\epsilon$ to compute  $p \in K$, $p' \in K'$  so that  $d(p,p') \leq \epsilon \rho_* $ and the total arithmetic complexity of the algorithm  satisfy, respectively
\begin{equation} \label{iter1}
k_\epsilon =  O \bigg (\frac{1}{\epsilon^2} \bigg ), \quad
O \bigg (\frac{T}{\epsilon^2} \bigg ).
\end{equation}
(ii) Suppose $\delta_*=d(K,K') >0$. Let $\rho_*$ be the maximum of the diameters of $K$ and $K'$. The number of iterations $k_{\delta_*}$ to compute a witness pair and the total arithmetic complexity satisfy, respectively,
\begin{equation} \label{iter2}
k_{\delta_*}=
O \bigg  ( \frac{\rho_*^2}{\delta_*^2}\bigg ), \quad
O \bigg  ( \frac{T\rho_*^2}{\delta_*^2}\bigg ).
\end{equation}
\end{thm}

\begin{proof} From Lemma \ref{lem1} and definition of $k(\delta_0)$ (see (\ref{iter})), in order to halve the initial gap from $\delta_0$ to $\delta_0/2$, in the worst-case  Triangle Algorithm I requires $k(\delta_0)$ iterations.  Then, in order to reduce the gap from $\delta_0/2$ to $\delta_0/4$ it requires at most $k(\delta_0/2)$ iterations, and so on. From (\ref{iter}), for each nonnegative integer $r$  the worst-case number of iterations to reduce a gap from
$\delta_0/2^r$ to  $\delta_0/2^{r+1}$ is given by
\begin{equation}
k \bigg (\frac{\delta_0}{2^r} \bigg ) \leq   \bigg \lceil N \bigg (\frac{\delta_0}{2^r} \bigg ) \bigg  \rceil = \lceil 2^{2r} N_0 \rceil \leq 2^{2r}  \lceil N_0  \rceil.
\end{equation}
Therefore, if $t$ is the smallest index such that $\delta_0/2^{t} \leq \epsilon \rho_*$,  i.e.
\begin{equation}
2^{t-1} < \frac{\delta_0}{\rho_* \epsilon} \leq  2^t,
\end{equation}
then the total number of iterations of the algorithm, $k_\epsilon$,  to test if condition (i) is valid satisfies:
\begin{equation}
k_\epsilon \leq  \lceil N_0  \rceil (1+2^2+2^4 + \dots + 2^{2(t-1)}) \leq   \lceil N_0  \rceil  \frac{2^{2t} -1}{3} \leq   \lceil N_0  \rceil  2 \times 2^{2(t-1)} \leq (N_0+1) \frac{2\delta^2_0}{ \rho_*^2 \epsilon^2}.
\end{equation}
From (\ref{iter}) we get

\begin{equation}
k_\epsilon \leq \bigg (23 \frac{ \rho_*^2}{\delta^2_0} +1 \bigg) \frac{2\delta^2_0}{ \rho_*^2 \epsilon^2} = \bigg (23 +  \frac{\delta^2_0}{ \rho_*^2} \bigg ) \frac{2}{ \epsilon^2}.
\end{equation}
Since $K \cap K' \not = \emptyset$ and from the definition of $\rho_*$,  $\delta_0 \leq \rho_*$, hence we get the claimed bound on $k_\epsilon$  in (\ref{iter1}).

Suppose $\delta_* >0$. It suffices to choose
$$\epsilon = \frac{\delta_*}{2\rho_*}.$$
Then in
$$k
_\epsilon \leq \frac{48}{\epsilon^2}= \frac{192\rho^2_*}{\delta^2_*}$$
iterations we can determine that $\delta_* >0$.
\end{proof}

\section{Algorithm for Approximation of Distance and Optimal Support}

Our goal in this section and next is to start with a witness pair $(p,p')$, $p \in K$, $p' \in K'$, then continue to iterate to get new witnesses that would estimate $\delta_*=d(K,K')$ to within a prescribed error, or a pair of supporting hyperplanes that would estimate the optimal pair to within a prescribed tolerance. The following is easy to prove.

\begin{prop}  \label{propnew} Given a  pair $(p,p') \in K \times K'$, the orthogonal bisector hyperplane of the line segment $pp'$ is
\begin{equation}
H=\{x \in \mathbb{R} ^m:  \quad h^Tx = a \},
\quad h=p-p', \quad a = \frac{1}{2} (p^Tp-p'^Tp').
\end{equation}
If $(p,p')$ is a witness pair then
\begin{equation}
K \subset H_{+}=\{x \in \mathbb{R} ^m:  h^Tx >  a  \}, \quad K' \subset   H_{-}=\{x \in \mathbb{R}^m:  h^Tx < a \}. ~~~\Box
\end{equation}
\end{prop}

The following theorem shows that once we have a witness pair $(p,p')$, by solving two convex programming problems not only do we obtain a lower bound to $\delta_*=d(K,K')$, but also a pair of supporting hyperplanes that are parallel to the orthogonal bisecting hyperplane of the line segment $pp'$. Given that $d(p,p')$ is an upper bound on $\delta_*$, the difference between $d(p,p')$ and the lower bound gives a measure of how well the current witness pair estimates $\delta_*$.

\begin{thm} \label{thm6}
Suppose $(p,p') \in K \times K'$ is a witness pair {\rm (see Figure \ref{Fig6AA})}.  Let $H$ be the orthogonal bisecting hyperplane to the line $pp'$, thus  $H=\{x: h^Tx= a \}$, where $h=p-p'$, $a = \frac{1}{2} (\Vert p \Vert^2- \Vert p'\Vert^2 )$. Let
\begin{equation} \label{vv'}
v={\rm argmin}\{h^Tx: x \in K\},  \quad v'={\rm argmax}\{h^Tx: x \in K'\}.
\end{equation}
Let
\begin{equation} \label{hh}
H_v= \{x: h^Tx= h^Tv\}, \quad H_{v'}= \{x: h^Tx= h^Tv'\}.
\end{equation}
Then the hyperplanes $H_v$ and $H_{v'}$ give supporting hyperplanes to $K$ and $K'$, respectively, and the distance between them, $d(H_v, H_{v'})$ is a lower bound to $\delta_*$.
Specifically, let
\begin{equation} \label{deltavv}
\delta_v=d(v, H), \quad  \delta_{v'}= d(v', H).
\end{equation}
Then if
\begin{equation} \label{deltalow}
\underline \delta= \delta_v + \delta_{v'},
\end{equation}
we have
\begin{equation} \label{maineq}
d(H_v, H_{v'})=\underline \delta= \frac{h^Tv-h^Tv'}{\Vert h \Vert},
\end{equation}
and
\begin{equation} \label{maineq2}
\underline \delta \leq \delta_* \leq \delta=d(p,p').
\end{equation}
\end{thm}
\begin{proof} It should be clear that the optimization that defines $v$ in (\ref{vv'}) implies that $v$ is a point of $K$ closest to $H$.   Similarly, the optimization that defines $v'$ in (\ref{vv'}) implies that $v'$ is a point in $K'$ closest to $H$. Figure \ref{Fig6AA} gives a 2D depiction of these. In the figure the lines $H$, $H_v$, $H_{v'}$ actually depict  orthogonal hyperplanes to the line $pp'$.  However, it should be noted that in higher dimensions the points $p,p',v,v'$ are not necessarily coplanar.

Since $v \in K$ and $v' \in K'$,  from the definition of $H$ and its properties given in Proposition \ref{propnew}, it follows that $h^Tv > a $ and $h^Tv' < a$. The distance between $v$ and $H$ and the distance between $v'$ and $H$ can be calculated to be
\begin{equation}
\delta_v= \frac{h^Tv - a}{\Vert h \Vert}, \quad  \delta_{v'}=\frac{a - h^Tv'}{\Vert h \Vert}.
\end{equation}
Thus from (\ref{deltavv}),  (\ref{deltalow}) and  (\ref{maineq}) we get,
\begin{equation}
d(H_v, H_{v'}) = \underline \delta = \frac{h^Tv - h^Tv'}{\Vert h \Vert}.
\end{equation}
The fact that $\underline \delta$ is a lower bound to $\delta_*$ is obvious because $H_v$ and $H_{v'}$ give supporting hyperplanes to $K$ and $K'$, respectively.
\end{proof}

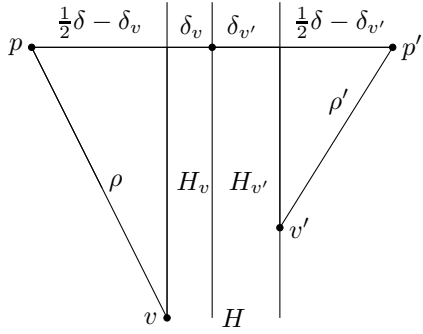
\begin{figure}[htpb]
	\centering
	\begin{tikzpicture}[scale=0.6]

      \draw (-4,0) -- (4,0) node[pos=0.55, above] {};
      \draw (0,0) -- (4,0) node[pos=0.71, above] {$\frac{1}{2}\delta-\delta_{v'}$};
      \draw (-4,0) -- (-1,0) node[pos=0.5, above] {$\frac{1}{2}\delta-\delta_v$};
      \draw (-1,0) -- (0,0) node[pos=0.55, above] {$\delta_v$};
      \draw (1.5,0) -- (0,0) node[pos=0.55, above] {$\delta_{v'}$};
      \draw (-4,0) -- (-2.44,-3.12);
      \draw (0,-6) -- (0,1) node[pos=0.5, right] {};
      \draw (1.5,1) -- (1.5,-4) node[pos=0.5, right] {};
      \draw (-1,-6) -- (-1,1) node[pos=0.3, left] {};
      \draw (4,0) -- (1.5,-4) node[pos=0.3, left] {$\rho'$};
      \draw (-1,0) -- (-1,-6) node[pos=0.5, right] {$H_v$};
      \draw (1.5,0) -- (1.5,-6) node[pos=0.5, left] {$H_{v'}$};
       \filldraw (-4,0) circle (2pt);
       \filldraw (4,0) circle (2pt);
       \filldraw (-1,-6) circle (2pt);
		\draw (-4,0) node[left] {$p$};
		\draw (4,0) node[right] {$p'$};
        \draw (-1,-6) node[left] {$v$};
        \draw (0,-6) node[right] {$H$};
        \draw (1.5,-4) node[right] {$v'$};
         \filldraw (1.5,-4) circle (2pt);
         \draw (-4,0) -- (-1,-6) node[pos=0.5, right] {$\rho$};

           \filldraw (0,0) circle (2pt);
           \draw (0,0) node[below] {};
		
	\end{tikzpicture}
\begin{center}
\caption{Depiction of the orthogonal bisector hyperplane $H$ to $pp'$, and parallel supporting hyperplanes $H_v$ and $H_{v'}$ that separate $K$ and $K'$.} \label{Fig6AA}
\end{center}
\end{figure}

\begin{cor} \label{cor1} Suppose $(p,p')$ is a witness pair. If
\begin{equation}
p={\rm argmin}\{h^Tx: x \in K\}, \quad {\rm or} \quad  p'={\rm argmax}\{h^Tx: x \in K'\},
\end{equation}
then
\begin{equation} \label{halfbound}
\frac{1}{2}d(p,p') \leq \delta_* \leq d(p,p').
\end{equation}
In particular, if one of the two sets $K$ or $K'$ is a singleton element, then any witness pair satisfies (\ref{halfbound}).
If both
\begin{equation}
p={\rm argmin}\{h^Tx: x \in K\}, \quad p'={\rm argmax}\{h^Tx: x \in K'\},
\end{equation}
then
\begin{equation}
d(p,p')=\delta_*.
\end{equation}
\end{cor}

\begin{proof} The proof of  both  parts follow from the fact that $\underline  \delta \leq \delta_*$.
\end{proof}

\begin{remark} Note that $p={\rm argmin}\{h^Tx: x \in K\}$ if and only if there exists no weak-pivot at $p$.  Similarly $p'={\rm argmax}\{h^Tx: x \in K'\}$ if and only if there is no weak-pivot at $p'$.
\end{remark}

\begin{definition} \label{def11} Let  $(p,p')$ be a witness pair and consider $v, v'$  as defined in (\ref{vv'}). Also let $H_v$, $H_{v'}$ and  $H$ be the hyperplanes defined earlier. Let $\delta=d(p,p')$,
$\rho=d(p,v)$, and $\rho'=d(p',v')$. Let
$\underline \delta=d(H_v, H_{v'})$,  and $E=\delta - \underline \delta$
(see Figure \ref{Fig6AA}). We shall say $(p,p')$ gives a {\it strong}  $\epsilon$-{\it approximate solution} to $\delta_*$ if either
\begin{equation}
E \leq \epsilon \rho, \quad {\rm or} \quad E \leq \epsilon \rho'.
\end{equation}
\end{definition}

\begin{prop} If $(p,p')$ gives a  strong  $\epsilon$-approximate solution to $\delta_*$, then
\begin{equation}
\delta- \delta_* \leq \epsilon \rho \quad {\rm or} \quad \delta- \delta_* \leq \epsilon \rho',
\end{equation}
i.e. $(p,p')$ is an $\epsilon$-approximate solution to $\delta_*$. Furthermore, $(H_v, H_{v'})$ is a pair of parallel supporting hyperplanes orthogonal to $pp'$ such that
\begin{equation}
\delta_* - d(H_v, H_{v'}) \leq \epsilon d(p,p'),
\end{equation}
i.e. $(p,p')$ is an $\epsilon$-approximate solution to the supporting hyperplanes problem.
\end{prop}

\begin{proof} The proof is immediate from the inequalities $\underline  \delta \leq \delta_* \leq \delta$.
\end{proof}

Theorem \ref{thm6} and its corollary suggest an algorithm for computing an initial approximation to $\delta_*=d(K,K')$ when it is positive: Run  Triangle Algorithm I to compute a witness pair $(p,p')$.  Next compute the error
\begin{equation}
E= \delta - \underline \delta.
\end{equation}
If $(p,p')$ is an $\epsilon$-approximate solution to $\delta_*$, stop. Otherwise,  we need to improve the gap, $d(p,p')$.  While a pivot is no longer available,  we will show that when either $E > \epsilon \rho$, or $E >  \epsilon \rho'$, we can still make use of $v$ or $v'$ defined in (\ref{vv'}) in order to compute the nearest point to $p$ on the line $p'v'$, or the nearest point to $p'$ on the line $pv$.  Then we run Triangle Algorithm I until a new  witness pair is obtained and the process is repeated. In summary, Triangle Algorithm II computes a strong $\epsilon$-approximate solution to $\delta_*$.  These will be formalized in the next section.


\section{Triangle Algorithm II: Properties and Complexity Analysis}
We first give a definition.

\begin{definition} \label{def12} Given a witness pair $(p,p')$, let $\delta = d(p,p')$ and $\delta_v$, $\delta_{v'}$ be as defined in (\ref{deltavv}).  Define
\begin{equation}  \label{Ev}
E_v= (\frac{1}{2} \delta - \delta_v), \quad E_{v'}= (\frac{1}{2}  \delta - \delta_{v'}).
\end{equation}
Clearly,
\begin{equation}
E= \delta - \underline \delta = E_{v} + E_{v'}.
\end{equation}
\end{definition}

Consider the following algorithm, called \emph{Triangle Algorithm II}. Its input is a witness pair $(p,p')$. It computes a new witness pair $(p,p')$ that gives an $\epsilon$-approximate solution to $\delta_*$, as well as a pair $(v,v') \in K \times K$, where the hyperplanes  parallel to the orthogonal bisecting hyperplane of $pp'$ passing through $v,v'$ form a pair of  supporting hyperplanes, giving  an $\epsilon$-approximate solution to the supporting hyperplanes problem.

\begin{center}
\begin{tikzpicture}
\node [mybox] (box){%
    \begin{minipage}{0.9\textwidth}
{\bf  Triangle Algorithm II ($(p,p') \in K \times K'$, a witness pair, $\epsilon \in (0,1)$)}\

{\bf Step 1.} Set $h=p-p'$,  $\delta=d(p,p')$,  $a = \frac{1}{2} (\Vert p \Vert^2- \Vert p'\Vert )^2$.  Compute
$$v={\rm argmin}\{h^Tx: x \in K\}, \quad v'={\rm argmax}\{h^Tx: x \in K'\}.$$
Set
$$\underline \delta= \frac{h^Tv-h^Tv'}{\Vert h \Vert}, \quad  E= \delta- \underline \delta, \quad
\delta_v= \frac{h^Tv - a}{\Vert h \Vert}, \quad  \delta_{v'}=\frac{a - h^Tv'}{\Vert h \Vert}.$$

{\bf Step 2.} If $E \leq  \epsilon \rho$,  or
$E\leq  \epsilon \rho'$, with $\rho=d(p,v)$, $\rho'=d(p',v')$, output $(p,p')$, $(H_v,H_{v'})$, stop.

{\bf Step 3.}  If $E_v= (\frac{1}{2} \delta - \delta_v) > \frac{1}{2} \epsilon \rho $, compute $p \leftarrow nearest(p'; pv)$, go to Step 5.

{\bf Step 4.}  If $E_{v'}= (\frac{1}{2}  \delta - \delta_{v'} > \frac{1}{2}\epsilon \rho'$,  compute  $p' \leftarrow nearest(p; p'v')$, go to Step 5.


{\bf Step 5.} Call Triangle Algorithm I with $(p, p')$ as input. Go to Step 1.


    \end{minipage}};
\end{tikzpicture}
\end{center}

\begin{remark} \label{rem2}  Note that if the algorithm does not terminate at Step 1,  then either Step 2 or Step 3 is executed.
\end{remark}

\begin{remark} It is possible that after a single iteration in Triangle Algorithm II the new pair $(p,p')$ is  no longer a witness pair. Thus the algorithm may need to call Triangle Algorithm I.  However, as will be proven, in each iteration of Triangle Algorithm II and each iteration within a call to Triangle Algorithm I, we are able to reduce the gap $d(p,p')$ at the current witness pair by a certain reasonable amount.
\end{remark}

\begin{remark} \label{rem3} Since $(p,p')$ is a witness pair, the weak $p'$-pivot $v$ and $p$ must lie on the same side of $H$. In other words, the hyperplane parallel to $H$ passing through $v$, $H_v$, separates $p$ from $p'$ (see Figure \ref{Fig7A}).
\end{remark}

\begin{lemma} \label{lem2} Suppose that $(p,p')$ is a witness pair. Let $v$, $\delta$, $\delta_v$, $E_v$ and $H_v$ be as defined previously (see (\ref{Ev})).   Let $\overline v$ be the intersection of $H_v$ and $pp'$. Let $\gamma=d(v, \overline v)$. Let $q= nearest(p';pv)$.  Let $\delta'=d(q,p')$, and $\rho=d(p,v)$.

(i) Suppose  $\angle pqp' = \pi/2$ {\rm (see Figure \ref{Fig8})}. Then
$$\delta' = \delta \sqrt{1- \frac{E^2_v}{\rho^2}}.$$

(ii) Suppose $\angle pqp' > \pi/2$ (i.e. $q=v$) {\rm (see Figure \ref{Fig8A})}. Then
$$\delta' \leq \delta \sqrt{1- \frac{E^2_v}{\delta^2}}.$$

\end{lemma}
\begin{proof}
The right triangles $\triangle pqp'$ and  $\triangle p\overline v v$ have the angle $\angle qpp'$ in common and therefore congruent (see Figure \ref{Fig8}). We may write
$$\frac{\delta'}{\delta}= \frac{\gamma}{\rho}.$$
Substituting for $\gamma= \sqrt{\rho^2 -E^2_v}$ we get
$$\frac{\delta'}{\delta}= \frac{\sqrt{\rho^2 -E^2_v}}{\rho}.$$
Hence the proof of (i).

To prove (ii), since $\angle pqp'$ is obtuse, it is easy to argue that
$$\delta'^2+ \rho^2 \leq \delta^2.$$
Since $E_v \leq \rho$, we have
$$\delta'^2\leq \delta^2 - E_v^2.$$
Hence the proof of (ii).
\end{proof}
\begin{remark}  \label{rem4} An identical result to Lemma \ref{lem2} can be stated for the case when $E_{v'} \geq \epsilon \rho'$.
\end{remark}

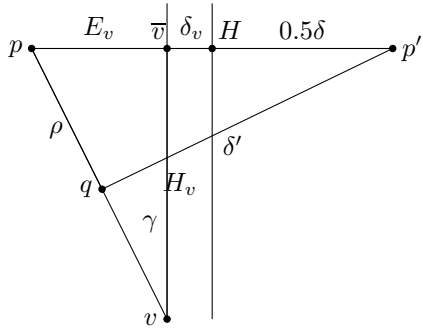
\begin{figure}[htpb]
	\centering
	\begin{tikzpicture}[scale=0.6]

      \draw (-4,0) -- (4,0) node[pos=0.55, above] {$H$};
      \draw (0,0) -- (4,0) node[pos=0.5, above] {$0.5\delta$};
      \draw (-4,0) -- (-1,0) node[pos=0.5, above] {$E_v$};
      \draw (-1,0) -- (0,0) node[pos=0.55, above] {$\delta_v$};
      \draw (4,0) -- (-2.44,-3.12) node[pos=0.55, below] {$\delta'$};
      \draw (-4,0) -- (-2.44,-3.12);
      \filldraw (-2.44,-3.12) circle (2pt);
      \draw (-2.44,-3.12) node[left] {$q$};
      \draw (0,-6) -- (0,1) node[pos=0.5, right] {};
      \draw (-1,-6) -- (-1,1) node[pos=0.3, left] {$\gamma$};
       \filldraw (-4,0) circle (2pt);
       \filldraw (4,0) circle (2pt);
       \filldraw (-1,-6) circle (2pt);
		\draw (-4,0) node[left] {$p$};
		\draw (4,0) node[right] {$p'$};
        \draw (-1,-6) node[left] {$v$};
         \draw (-1,0) node[above] {$\overline v~~$};
         \filldraw (-1,0) circle (2pt);
         \draw (-4,0) -- (-1,-6) node[pos=0.3, left] {$\rho$};

           \filldraw (0,0) circle (2pt);
           \draw (-1,0) -- (-1,-6) node[pos=0.5] {$~~~H_v$};
		
	\end{tikzpicture}
\begin{center}
\caption{Closest point, $q$, on line segment $pv$ is strictly interior.} \label{Fig8}
\end{center}
\end{figure}

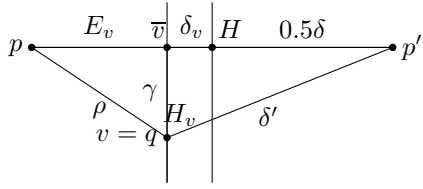
\begin{figure}[htpb]
	\centering
	\begin{tikzpicture}[scale=0.6]

      \draw (-4,0) -- (4,0) node[pos=0.55, above] {$H$};
      \draw (0,0) -- (4,0) node[pos=0.5, above] {$0.5\delta$};
      \draw (-4,0) -- (-1,0) node[pos=0.5, above] {$E_v$};
      \draw (-1,0) -- (0,0) node[pos=0.55, above] {$\delta_v$};
      \draw (0,-3) -- (0,1) node[pos=0.5, right] {};
      \draw (-1,-3) -- (-1,1) node[pos=0.5, left] {$\gamma$};
       \filldraw (-4,0) circle (2pt);
       \filldraw (4,0) circle (2pt);
       \filldraw (-1,-2) circle (2pt);
		\draw (-4,0) node[left] {$p$};
		\draw (4,0) node[right] {$p'$};
        \draw (-1,-2) node[left] {$v=q$};
         \draw (-1,0) node[above] {$\overline v~~$};
         \filldraw (-1,0) circle (2pt);
         \draw (-4,0) -- (-1,-2) node[pos=0.5, below] {$\rho$};

           \filldraw (0,0) circle (2pt);
           \draw (-1,0) -- (-1,-3) node[pos=0.5] {~~~$H_v$};
           \draw (4,0) -- (-1,-2) node[pos=0.55, below] {$\delta'$};
		
	\end{tikzpicture}
\begin{center}
\caption{Closest point on line segment $pv$ is $v$.}  \label{Fig8A}
\end{center}
\end{figure}

In the next theorem all quantities are as defined earlier.

\begin{thm} \label{thm7} Suppose $(p,p')$ is a witness pair and $E_v \geq \epsilon \delta/2$.  Then we have

\begin{equation} \label{gap2}
\delta' \leq
\begin{cases}
\delta \sqrt{1- \frac{\epsilon^2 \delta_*^2}{4\rho^2_*}}, &\text{if $\rho \geq \delta $;}\\
\\
\delta \sqrt{1- \frac{\epsilon^2}{4}}, &\text{if $\rho < \delta $.}\\
\end{cases}
\end{equation}
\end{thm}
\begin{proof} We consider the two cases separately.

{\bf Case I.} $\rho \geq \delta$. In this case $\angle pqp'$ is a right angle. Thus the first case of Lemma \ref{lem2} applies
$$\frac{\delta'}{\delta}= \frac{\sqrt{\rho^2 -E^2_v}}{\rho}.$$
We have
$$E_v \geq \frac{\epsilon \delta}{2} \geq \frac{\epsilon \delta_*}{2}.$$
Using the above and since $\rho \leq \rho_*$ we get the proof of this case.

{\bf Case II.} $\rho <\delta$.  We break this up into two subcases.

{\bf Subcase 1.} $\angle pqp' =\pi/2$. In this subcase, case (ii) of  Lemma \ref{lem2} implies
$$\delta' = \delta \sqrt{1- \frac{E^2_v}{\rho^2}}.$$
We have
$$\frac{E^2_v}{\rho^2} \geq \frac{E^2_v}{\delta^2} \geq \frac{\epsilon^2}{4\delta^2}.$$
It follows that
$$\delta' \leq \delta \sqrt{1- \frac{\epsilon^2}{4}}.$$

{\bf Subcase 2.}  $\angle pqp' > \pi/2$. From Lemma \ref{lem2} we have
$$\delta' \leq \delta \sqrt{1- \frac{E^2_v}{\delta^2}}.$$
But using that $E_v \geq \epsilon \delta$ we get the desired result.

\end{proof}

\begin{remark}  \label{rem5} An identical result to Theorem \ref{thm7} can be stated for the case when $E_{v'} \geq \epsilon \rho'$.
\end{remark}

We may now state the complexity bounds, recalling $\rho_0$ (initial gap), $\rho_*$ (maximum of the diameters of $K$ and $K'$), $T$ maximum of complexity in solving a linear program over $K$ or $K'$ to get a pivot
(see (\ref{eqa5})), and $\delta_*=d(K,K')$.

\begin{thm} \label{thm8} Let $N_1$ be the number of times Triangle Algorithm II makes use of a weak pivot computed in Step 1 to reduce the gap in Step 2 or Step 3, and let $N_2$  the number of times it makes use of a pivot computed in a call to Triangle Algorithm I in Step 5. Then,
\begin{equation} \label{thm8eq1}
N_1+N_2= O \bigg  ( \frac{\rho_*^2}{\delta_*^2\epsilon^2} \ln \frac{\delta_0}{\delta_*}\bigg )=
O \bigg  ( \frac{\rho_*^2}{\delta_*^2\epsilon^2} \ln \frac{\rho_*}{\delta_*}\bigg ).
\end{equation}
Thus the total arithmetic complexity of Triangle Algorithm II is
\begin{equation} \label{thm8eq2}
O \bigg  ( T\frac{\rho_*^2}{\delta_*^2\epsilon^2} \ln \frac{\delta_0}{\delta_*}\bigg )=
O \bigg  ( T\frac{\rho_*^2}{\delta_*^2\epsilon^2} \ln \frac{\rho_*}{\delta_*}\bigg ).
\end{equation}
In particular, when one of the two sets $K$ or $K'$ is a singleton element we have
\begin{equation} \label{thm8eq3}
N_1+N_2=O \bigg  (\frac{\rho_*^2}{\delta_*^2\epsilon^2}.\bigg ).
\end{equation}
\end{thm}
\begin{proof}  In the worst-case in each contribution to $N_1$ we have

\begin{equation}
\delta' \leq \delta \sqrt{1- \frac{\epsilon^2 \delta_*^2}{4\rho^2_*}} \leq \delta \exp(-\frac{\epsilon^2 \delta_*^2}{8\rho^2_*}).
\end{equation}
So after $N_1$ iterations we may write
\begin{equation}
\delta_{N_1} \leq \delta_0  \exp(-{N_1}\frac{\epsilon^2 \delta_*^2}{8\rho^2_*}).
\end{equation}
By bounding the right-hand-side of the above to be less than or equal to $\delta_*$,  we get
\begin{equation} \label{N1}
N_1  \leq  \frac {8\rho^2_*} {\epsilon^2 \delta_*^2} \ln \frac{\delta_0} {\delta_*}.
\end{equation}
But $\delta_0$ is bounded above by $2\rho_* + \delta_*$. This is easy to see because if $\delta_*=d(p_*,p'_*)$, then the path that connects $p_0$ to $p'_0$ by connecting $p_0$ to $p_*$ to $p'_*$ to $p'_0$ has length at most
$2\rho_* + \delta_*$.

The number $N_2$ is also bounded by the same quantity. This follows from the analysis of Triangle Algorithm I already described. These imply
(\ref{thm8eq1}) and (\ref{thm8eq2}).

To prove (\ref{thm8eq3}) for the special case, we use the fact that Triangle Algorithm II begins with a witness pair. Thus by Corollary \ref{cor1} it follows that $\delta_0  \leq 2 \delta_*$.  If follows that $\ln \delta_0/\delta_*$ is a constant. Substituting in (\ref{N1}),  proof for the special case follows.
\end{proof}

\section{Complexity of Triangle Algorithms in Special Cases}

In this section we consider several important special cases and analyze the complexity of solving the approximation problems.

\subsection{The Case of A Finite Convex Hull and A Singleton Point}

Consider when $K= conv(V)$, $V=\{v_1, \dots, v_n\}$, $K'=\{p'\}$. In this case $\rho_*$ is the diameter of $V$:

\begin{equation}  \label{sec8eq1}
\rho_* = \max \{d(v_i,v_j): 1 \leq i \leq n\}.
\end{equation}

For this special case the complexity of computing a pivot
at a given point $x$ in $K$ (in Step 2 or Step 3 of Triangle Algorithm I) is $O(mn)$ arithmetic operations.  The complexity of computing a weak-pivot at a given witness point $x$ in $K$ (in Step 1 of Triangle Algorithm II) is also $O(mn)$ arithmetic operations. Thus the total number of arithmetic operations in Triangle Algorithm I to get an $\epsilon$-approximate solution when $\delta_*=0$, that of finding a witness pair, and the total number of arithmetic operations in Triangle Algorithm II to get an $\epsilon$-approximate solution to $\delta_*$ as well as to the supporting hyperplanes problem are, respectively

\begin{equation}  \label{sec8eq2}
O\bigg (\frac{mn}{\epsilon^2} \bigg)  {\rm ~(intersection)}, \quad
O \bigg  (\frac{mn\rho_*^2}{\delta_*^2} \bigg ) {\rm ~(separation)}, \quad
O \bigg ( mn \frac{\rho_*^2} {\delta_*^{2} }\frac{1}{\epsilon^{2}} \ln \frac {\rho_*}{\delta_*}  \bigg ) {\rm ~(distance~ \& ~support)}.
\end{equation}

The first complexity above is that of testing approximately if a point lies in the convex hull of $n$ points. This is a special case of linear programming.  The last one is the complexity of approximating the distance between a point and the convex hull of $n$ points when the point lies outside of the convex hull as well as the support. The results in particular give a more general version of the Triangle Algorithm in \cite{kal14} in the sense that it also computes the closest point in $K$ to $p'$ when $p'$ is not in $K$ and the support.

In \cite{kalSaks} we have shown that with an $O(mn^2)$-time preprocessing,  the complexity of each iteration can essentially  be reduced to $O(m+n)$ as opposed to $O(mn)$. Thus we may state the following corresponding alternate complexities to those of (\ref{sec8eq2}):

\begin{equation}  \label{sec8eq3}
O\bigg ((m +n )\frac{1}{\epsilon^2} \bigg), \quad
O \bigg  ((m +n ) \frac{\rho_*^2}{\delta_*^2} \bigg ), \quad
O \bigg ((m +n ) \frac{\rho_*^2} {\delta_*^{2} } \frac{1}{\epsilon^{2}} \ln \frac {\rho_*}{\delta_*}  \bigg ).
\end{equation}

The preprocessing time in \cite{kalSaks} consists of computing the distances $d(v_i,v_j)$, for all pairs $i,j$.  This is also useful in a version of the Triangle Algorithm where the coordinates of the singleton point $p' \in K'$ is unknown, see {\it Blindfold Triangle Algorithm} \cite{Kalan12}. Such a case of the problem may arise if we wish to test the feasibility of a site within the convex hull of known sites to lie within a prescribed distances from them.

Here we show that a different preprocessing is possible resulting in $O(m+n)$ complexity per iteration of Triangle Algorithm I or II. Specifically, consider finding a pivot at a given point $p$. This can be established by solving the following problem, see (\ref{eqa5})
\begin{equation}  \label{sec8eq4}
\max \{(p'-p)^Tv:  v = \sum_{i=1}^n x_i v_i, \quad \sum_{i=1}^n x_i=1, \quad x_i \geq 0\}.
\end{equation}
Denote by $A$ the $m \times n$ matrix $[v_1, \dots, v_n]$.  Since $p$ is an iterate in $K$, $p=A y$, for some $y =(y_1, \dots, y_n)^T$, $\sum_{i=1}^n y_i =1$, $y \geq 0$.  Assume  that we have computed a representation of $p'$ as a linear combination of $v_i$'s, not necessarily a convex combination, ay $p'=A y'$. This is easy to compute.
Thus we can write
\begin{equation}  \label{sec8eq5}
(p'-p)= A z, \quad z = y' - y.
\end{equation}
The objective function in (\ref{sec8eq4}) is $zTA^TA x$. It is easy to see the optimal value is simply the maximum coordinate of the vector $A^TA z$, say occurring at $x=e_j$ for some $j$.  This corresponds to $v_j$ as the pivot.  The vector $A^TA z$ can be computed in $O(mn)$ time. Another way of computing  $A^TA z$ is to compute $Q=A^TA$ in $O(mn^2)$ time, followed by computing  $Q z$ in $O(n^2)$ time.  This approach is less efficient, however there are overall advantages. Suppose the next iterate is  $\overline p= \gamma p + (1- \gamma) v_j$,  for some $v_j \in V$ and $\gamma \in (0,1)$. Thus
\begin{equation}  \label{sec8eq6}
\overline p= \gamma A y + (1 - \gamma) v_j= A (\gamma y + (1 - \gamma) e_j).
\end{equation}
This can be written as $\overline p= A \overline y$, where $\overline y=\gamma y + (1 - \gamma) e_j$.  Thus if we set $\overline z = y' - \overline y$,  then $A \overline z=p' - \overline p$. Now to compute a pivot for $\overline p$  we need to compute the minimum of $A^TA \overline z$. We have
\begin{equation}
\overline z =y' - \overline y= y'-\gamma y - (1 - \gamma) e_j= \gamma(y'- y) + (1- \gamma)(y' - e_j).
\end{equation}
Thus
\begin{equation}
Q \overline z = \gamma Q z + (1- \gamma)Q(y' - e_j).
\end{equation}
Note that since $Qz$ is computed, if we have also computed $Qy'$, a vector that stays fixed since $p'$ is fixed, then $Q \overline z$ can be computed in $O(n)$ time.  We may thus formally state the following result justifying the claimed complexity in (\ref{sec8eq3}).

\begin{prop} \label{propcomp}  Assume that $Q=A^TA$ is computed. Assume also we have computed $p'=Ay'$ for some $y'$. If the initial iterate is
$p={\rm argmin}\{d(p',v_i): v_i \in V\}$. Then  each iteration of the Triangle Algorithms can be carried out in $O(m+n)$ time.  $\Box$
\end{prop}

The $O(m+n)$ complexity is essentially the complexity of finding a pivot at the current iterate $p \in conv(V)$  and computing the closest point to $p'$ on the line segment $pv$, i.e. $nearest(p';pv)$.

We can improve upon this by trying to find a {\it minimum-angle pivot} defined below (see Figure \ref{Fig2xx}).

\begin{definition}  \label{defxx} Given   $p \in K =conv(\{v_1, \dots, v_n\})$, we say $v \in K$ is a {\it minimum-angle} $p'$-{\it pivot} for $p$ if $v$ is $p'$-{\it pivot} for $p$ (i.e. $d(p,v) \geq d(p',v)$) and the angle $\theta=\angle p'p v$ is the smallest among all such pivots.
\end{definition}

The advantage in using a minimum-angle pivot is that the distance between
$p'$ and $p''=nearest(p',pv)$ will be the least, hence the best reduction in the gap can be achieved in that iteration. We prove:

\begin{prop} \label{propcomp2}
Assume that $Q=A^TA$ is computed. Assume also that $p'=Ay'$ is computed for some $y'$. If the initial iterate is
$p={\rm argmin}\{d(p',v_i): v_i \in V\}$. Then  each iteration of the Triangle Algorithms, using a minimum-angle pivot, can be carried out in $O(m+n)$ time.
\end{prop}

\begin{proof} Let $\theta= \angle p'pv$.  Let $a= d(p,p')$, $b=d(p,v)$, and $c=d(p',v)$.  Then from the law of cosines $c^2=a^2+b^2-2ab \cos \theta$. Thus
\begin{equation}
\sin^2 \theta = 1 - \frac{(a^2+b^2-c^2)^2}{4a^2b^2}.
\end{equation}
To compute a minimum-angle pivot, it suffices to compute
\begin{equation} \label{optsine}
\min \{\sin^2 \theta:  \quad v \in V,  \quad (p'-p)^T(v-p) \geq 0 \}.
\end{equation}
Suppose for a given  iterate $p$ we have computed $a,b,c$ for each $v$, and suppose $\overline p$ is the next iterate computed according to a minimum-angle pivot. In order to compute a minimum-angle pivot for $\overline p$ we need to solve the optimization in (\ref{optsine}), replacing $p$ with $\overline p$.  This requires computing the corresponding $\overline a, \overline b, \overline c$ for each $v \in V$.
By similar analysis as in Proposition \ref{propcomp} we can argue that for each $v$ these require only constant update. Thus in $O(n)$ time we can compute the next minimum-angle pivot.  Analogously, to compute the new iterate takes $O(m+n)$ time overall.
\end{proof}

\begin{figure}[htpb]
	\centering
	\begin{tikzpicture}[scale=0.4]

	
		\draw (0.0,0.0) -- (7.0,0.0) -- (-2.0,-4.0) -- cycle;
		\draw (0,0) node[left] {$p'$};
\draw (-2,-4) -- (7,0) node[pos=0.5, right] {$b$};
\draw (0,0) -- (7,0) node[pos=0.5, above] {$c$};
\draw (0,0) -- (-2,-4) node[pos=0.5, left] {$a$};
		\draw (7,0) node[right] {$v$};
		\draw (-2,-4) node[below] {$p$};
\draw (-1.2,-2.5) node[below] {$\theta$};
           \filldraw (0,0) circle (2pt);
\filldraw (7,0) circle (2pt);
\filldraw (-2,-4) circle (2pt);
		


	\end{tikzpicture}
\begin{center}
\caption{$v$ a $p'$-pivot for $p$.} \label{Fig2xx}
\end{center}
\end{figure}
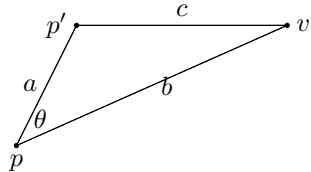

\begin{remark}  This problem can also be solved via the Frank-Wolfe method by minimizing $f(x)=\Vert Ax \Vert^2$ over the simplex $\{x: e^Tx=1, x \geq 0 \}$. Given an iterate $x_0$, Frank-Wolfe computes a line search by considering the minimum component of $\nabla f(x_0)= A^TAx_0$.   Vamsi  Potluru \cite{Vamsi} pointed out that (with preprocessing) computing the minimum component can be done in $O(m+n)$.
As mentioned in \cite{kalSaks}, via a distance computation to find a pivot in the Triangle Algorithm takes $O(m+n)$ time and here we have justified that finding a minimum-angle pivot which gives the best reduction at each iteration also takes $O(m+n)$ time. Indeed based on computational results the Triangle Algorithm seems to consistently outperform the Frank-Wolfe algorithm.  This is reported in the  relevant experimentations listed in the references, and will also be reported in forthcoming work, e.g. \cite{Hao}.  The present article has extended the potential utility of the Triangle Algorithm to much more general cases of the convex hull membership problem.
\end{remark}

\subsection{The Case of Two Finite Convex Hulls}
In this case we have a more general version of the previous case. We have
  $K= conv(V)$, $V=\{v_1, \dots, v_n\}$, $K'= conv(V')$, $V'=\{v_1', \dots, v_{n'}'\})$. The quantity $\rho_*$ can explicitly be computed:
\begin{equation}
\rho_* = \max \{d(v_i,v_j), d(v'_{i'}, v'_{j'}): 1 \leq i,j \leq n, 1 \leq i',j' \leq n'\}.
\end{equation}

The total number of arithmetic operations in Triangle Algorithm I to get an $\epsilon$-approximate solution when $\delta_*=0$, or a witness pair, or the
total number of arithmetic operations in Triangle Algorithm II to get an $\epsilon$-approximate solution to distance and support problems are, respectively

\begin{equation}
O\bigg ( m \max \{n, n'\} \frac{1}{\epsilon^2} \bigg), \quad
O \bigg  ( m \max \{n, n'\} \frac{\rho_*^2}{\delta_*^2} \bigg ), \quad
O \bigg ( m \max \{n, n'\} \frac{\rho_*^2} {\delta_*^{2} } \frac{1}{\epsilon^{2}} \ln \frac {\rho_*}{\delta_*}  \bigg ).
\end{equation}
Note that the dependence on $n$ and $n'$ of the complexities in  solving the desired problems in the case of two convex hulls is additive (i.e. $\max\{n,n'\}=O(n+n')$) as opposed to being multiplicative (i.e. $O(nn')$) which would have been the case, say for testing the approximate intersection, had we formulated the problem as that of finding the zero vector in the Minkowski difference.  In fact we claim that in this case too  with a preprocessing, the complexity of each iteration can be reduced to $O(m+\max\{n,n'\})$. To see this note that when neither set is a singleton,  the iterates $(p,p')$  are both changing. However,  this happens one at a time, i.e. either $p$ changes while $p'$ stays fixed, or the other way around. Thus in each iteration we can apply the same argument as in the previous case.  Thus we may state the following

\begin{thm} With a preprocessing time of $O(m\max \{n^2, n'^2\})$, the following alternate complexities can be stated
\begin{equation}  \label{sec82eq3}
O\bigg ((m +\max \{n, n'\} )\frac{1}{\epsilon^2} \bigg), \quad
O \bigg  ((m + \max \{n, n'\}  ) \frac{\rho_*^2}{\delta_*^2} \bigg ), \quad
O \bigg ((m + \max \{n, n'\}  ) \frac{\rho_*^2} {\delta_*^{2} } \frac{1}{\epsilon^{2}} \ln \frac {\rho_*}{\delta_*}  \bigg ). ~ \Box
\end{equation}

\end{thm}

\subsection{The Case of Two Explicit Polytopes}

Next consider the case with $K= \{x: Ax \leq b\}$,  $K'=\{x: A'x \leq b'\}$, where $A$ is $n \times m$ and $A'$ is $n' \times m$.  We assume both $K$ and $K'$ are bounded, hence polytopes.

In this case the complexity of computing a pivot, $T$,  is merely the complexity of solving a linear program. The quantity $\rho_*$ cannot be easily estimated but it can be approximated, e.g. by maximizing the one-norm of $x$ over $K$ and $K'$. However, in practice we do not need to estimate this value to run the algorithms. See Table 1 for the corresponding complexities.

\begin{remark}  To test if $K$ and $K'$ intersect we merely need to solve a feasibility problem: test if $\{x: Ax \leq b, A'x \leq b'\}$ is nonempty.  If $K$ and $K'$ are disjoint an application of Farkas Lemma reveals  that for some $w$,  $A^Tw=-A'^Tw$, $b^Tw < -b'^Tw$.  This in turn results in a separating hyperplane, see Theorem 17.3  Chv\'atal \cite{chvatal2}.   However, to compute the actual distance between them and the best supporting vector is not derivable via an LP. Those can be approximated via Triangle Algorithm I and II.
\end{remark}

\subsection{The Case of An Explicit Polytope and a Point}

Next consider the case with $K= \{x: Ax \leq b\}$,  where $A$ is $n \times m$ and  $K'=\{p'\}$, a single point. Clearly, to test feasibility of $p'$ in $K$ takes only $O(mn)$ operations.  When $p' \not \in K'$, also in $O(mn)$ operations we can compute a separating hyperplane.  However, to compute the closest point of $K$ we need to solve the quadratic programming problem:
$$\min \{d(p',x)^2:  Ax \leq b \}.$$
The complexities for computing the distance within a factor of two as well as those of computing an $\epsilon$-approximation to $\delta_*$ are given in Table 1.

\begin{remark}
Consider the quadratic program $\min \{x^TQx+c^Tx:  A x \leq b \}$, where we assume $\{x: Ax \leq b\}$ is bounded, and $Q$ is positive definite.
By adding an appropriate constant we can write the objective function as
$(x-x_0)^TQ(x-x_0)=x^TQx+ 2x_0^TQx + x_0^TQx_0$, where $Qx_0=\frac{1}{2}C$.
Letting $Q=B^TB$ be a Cholesky factorization,  the mapping  $y=B(x-x_0)$ results in a problem equivalent to computing the closest point to a polytope: $\min \{\Vert y \Vert^2:  A' y \leq b \}, \quad A'=AB^{-1}.$
Alternatively, it may be possible to restate the notion of pivot and distance dualities in terms of the metric $d_Q(p,q)=(p- q)^TQ(p-q)$.
\end{remark}

\section{Conclusions and Future Work}   In this article we have presented a new separating hyperplane theorem for two compact convex subsets  $K, K'$ of the Euclidean space.  The {\it distance dualities} proved here characterize when two such sets intersect and when they are disjoint. We used these dualities to give an algorithm for testing if the sets intersect in the form of computing a pair $(p,p') \in K \times K'$, where $d(p,p')$ is within a given prescribed tolerance. When the two sets are disjoint, the algorithm computes a separating hyperplane.  Triangle Algorithm I accomplishes these two tasks. The main work in each iteration is the search for a pivot which can be accomplished by solving a linear program over $K$ or $K'$. Thus the complexity of each iteration depends on the nature of description of the sets $K$ and $K'$. The computation of a pivot  does not necessarily require solving a full linear program over $K$ or $K'$.  Such is also the case for a weak-pivot. When the two sets are disjoint, Triangle Algorithm I terminates with a pair of points $(p,p') \in K \times K'$ where the orthogonal bisecting hyperplane separates $K$ and $K'$. Such a pair is called a {\it witness pair}. For such a pair, neither $K$ nor $K'$ admit a pivot.

Given a witness pair $(p,p')$, Triangle Algorithm II  begins to approximate $d(K,K')$, the distance between $K,K'$,  making use of weak-pivots to reduce the gap $d(p,p')$.  If needed it also resorts to Triangle Algorithm I. When Triangle Algorithm II has computed a pair $(p,p')$ with $d(p,p')$ within a prescribed tolerance of $d(K,K')$, it also leads to a pair of parallel supporting hyperplanes $(H,H')$ where $d(H,H')$ is within a prescribed tolerance of the optimal margin.

In the article, we also considered the complexity of these algorithms in some practical special cases. In particular, when $K$ or $K'$ are the convex hull of finite number of points. These find applications in such significant problems as linear and quadratic programming, SVM, and statistics.   Our computational results for the convex hull membership problem already supports Triangle Algorithm I as a promising algorithm that  seems to outperform such well known algorithms as Frank-Wolfe and the simplex method on tested problems (see \cite{Meng}, \cite{Gibson}, \cite{Hao}).

In a recent work, see \cite{Mayank}, we consider the problem of testing, for two ﬁnite sets of points in the Euclidean space, if their convex hulls are disjoint and computing an optimal supporting hyperplane if so. This is the case when $K$ and $K'$ are finite convex hulls, a fundamental problem of classiﬁcation in machine learning known as the {\it hard-margin SVM}. The problem can be formulated as a quadratic programming problem, see \cite{Vapnik1}, \cite{Vapnik2}.
 The SMO algorithm, see \cite{Platt},  is the current state of the art algorithm for solving this problem, however it does not answer the question of separability. An alternative to solving both problems is via phase I and II of  the Triangle Algorithm described in this article.  We have compared the performance of Triangle Algorithm with SMO for ﬁnding the optimal supporting hyperplane. Based on experimental results ranging up to $5000$ points in each set in dimensions up to $10000$, the triangle algorithm outperforms SMO.

While in this article we have only worked with pivots, the notion of {\it strict pivots}, proved for the convex hull membership problem in \cite{kal14}, can also be extended to the general cases considered here. Additionally, when $K$ is a general compact convex set and $K'$ a singleton centered at a ball of radius $\rho$ in the relative interior of $K$, more efficient complexities can be stated in terms of $\ln (1/\epsilon)$ and $\rho$, see \cite{kal14}.

There are interesting extensions and both theoretical and practical applications of these results. We enumerate a few.

{\bf (1)}  Extension of the Triangle Algorithm for the soft-margin SVM problem as well as the kernel version of both hard and soft-margin SVM problems.

{\bf (2)} Appropriate versions of the distance dualities and corresponding Triangle Algorithms can be stated for the case where the sets $K, K'$ are only assumed to be closed and convex, or only one is assumed to be compact. It is well-known that the separating hyperplane theorem is valid when only one of the closed convex sets is compact. Assuming that $(p_0,p_0') \in K \times K'$ is given, where $K,K'$ are only closed convex, we can reduce the problem of testing if they intersect to the case of compact convex sets considered here. We thus expect that some of our results extend to more general cases with ease.

{\bf (3)}  The distance dualities proved here for two compact convex sets extend to the case of more than two sets. Such a general version finds applications in machine learning.  Also, a version of the Triangle Algorithm can be extended to this case.

{\bf (4)}  Extension of the distance dualities and Triangle Algorithms to more general metric spaces, e.g. semidefinite programming.

{\bf (5)} Computing approximate solutions to NP-complete problems.

{\bf (6)} Combinatorial problems.

We will consider all the above aspects in future work. In summary, we expect wide applications of the results in convex and non-convex programming.

\bigskip


\end{document}